\documentclass[pdflatex,sn-mathphys-ay]{sn-jnl}

\usepackage{graphicx}%
\usepackage{multirow}%
\usepackage{amsmath,amssymb,amsfonts}%
\usepackage{amsthm}%
\usepackage{mathrsfs}%
\usepackage[title]{appendix}%
\usepackage{xcolor}%
\usepackage{textcomp}%
\usepackage{manyfoot}%
\usepackage{booktabs}%
\usepackage{algorithm}%
\usepackage{algorithmicx}%
\usepackage{algpseudocode}%
\usepackage{listings}%
\usepackage{tikz}
\usetikzlibrary{shapes.geometric, arrows}
\usetikzlibrary{mindmap}
\usetikzlibrary{positioning}
\tikzstyle{process} = [rectangle, minimum width=6cm, minimum height=1cm, text centered, text width=7cm, draw=black, fill=orange!30]
\allowdisplaybreaks
\makeatletter
\usepackage{bbm}
\usepackage{braket}
\usepackage{mathtools}
\DeclareMathOperator*{\argmax}{arg\,max}
\DeclareMathOperator*{\argmin}{arg\, min}
\allowdisplaybreaks
\makeatletter
\newenvironment{breakablealgorithm}
  {
   \begin{center}
     \refstepcounter{algorithm}
     \hrule height.8pt depth0pt \kern2pt
     \renewcommand{\caption}[2][\relax]{
       {\raggedright\textbf{\ALG@name~\thealgorithm} ##2\par}%
       \ifx\relax##1\relax 
         \addcontentsline{loa}{algorithm}{\protect\numberline{\thealgorithm}##2}%
       \else 
         \addcontentsline{loa}{algorithm}{\protect\numberline{\thealgorithm}##1}%
       \fi
       \kern2pt\hrule\kern2pt
     }
  }{
     \kern2pt\hrule\relax
   \end{center}
  }
\makeatother

\theoremstyle{thmstyleone}%
\newtheorem{theorem}{Theorem}
\newtheorem{fact}[theorem]{Fact}
\newtheorem{lemma}[theorem]{Lemma}

\theoremstyle{thmstyletwo}%

\theoremstyle{thmstylethree}%

\begin{document}

\title{Quantum Algorithm for Apprenticeship Learning}


\author*[1]{\fnm{Andris} \sur{Ambainis}}\email{andris.ambainis@lu.lv}

\author*[1]{\fnm{Debbie} \sur{Lim}}\email{limhueychih@gmail.com}
\equalcont{These authors contributed equally to this work.}

\affil[1]{Center for Quantum Computing Science, Faculty of Sciences and Technology, University of Latvia, Latvia}


\abstract{Apprenticeship learning is a method commonly used to train artificial intelligence systems to perform tasks that are challenging to specify directly using traditional methods. Based on the work of Abbeel and Ng (ICML'04), we present a quantum algorithm for apprenticeship learning via inverse reinforcement learning. As an intermediate step, we give a classical approximate apprenticeship learning algorithm to demonstrate the speedup obtained by our quantum algorithm. We prove convergence guarantees on our classical approximate apprenticeship learning algorithm, which also extends to our quantum apprenticeship learning algorithm. We also show that, as compared to its classical counterpart, our quantum algorithm achieves an improvement in the per-iteration time complexity by a quadratic factor in the dimension of the feature vectors $k$ and the size of the action space $A$.}




\maketitle

\section{Introduction}
Machine learning has emerged as a pivotal discipline, focusing on the development of algorithms and models that enable computers to learn from data and make decisions or predictions without explicit programming. This field encompasses a spectrum of approaches, such as supervised learning, unsupervised learning, and reinforcement learning, each suited to different learning scenarios and objectives~\citep{qiu2016survey,meng2020survey,fatima2017survey,muhammad2015supervised,jordan2015machine,baltruvsaitis2018multimodal,sen2020supervised,boutaba2018comprehensive,10.1145/3214306}. Apprenticeship learning is associated with the concept of learning by observing and imitating an expert. It consists of a range of techniques aimed at enabling autonomous agents to acquire skills and knowledge through interaction with a demonstrator or mentor. Unlike traditional supervised learning approaches, apprenticeship learning emphasizes learning from demonstration, where an agent seeks to replicate experts' behaviors by observing examples provided by a knowledgeable source. This learning framework has gained traction in various domains, including robotics and autonomous driving due to its ability to handle complex, real-world tasks and environments~\citep{abbeel2004apprenticeship, abbeel2008apprenticeship, coates2009apprenticeship, abbeel2005exploration, kolter2007hierarchical}. By leveraging apprenticeship learning, autonomous systems can effectively acquire and refine their capabilities, paving the way for more adaptive and versatile intelligent agents capable of performing tasks with human-like proficiency.

Reinforcement learning is another subfield of machine learning focused on training agents to make sequential decisions in dynamic environments by maximizing cumulative rewards. Interactions between the agent and the environment can be modelled as a Markov Decision Process (MDP)~\citep{puterman2014markov}. There are various MDP models with different assumptions that have been widely studied~\citep{wei2021learning,auer2008near,saldi2017asymptotic,neu2021online,woerner2015approximate,kara2023q,gao2021provably,altman2021constrained,guestrin2003efficient}. A closely related framework to reinforcement learning is inverse reinforcement learning. While an agent focuses on learning optimal policies from explicit rewards in reinforcement learning, an agent in inverse reinforcement learning seeks to infer the underlying reward function from observed demonstrations or expert trajectories. By leveraging on inverse reinforcement learning, agents can learn to imitate human-like decision-making processes and preferences, even in complex and uncertain environments. This approach is particularly useful in scenarios where the reward function is not explicitly provided, allowing for more flexible and adaptive learning strategies~\citep{fu2017learning,ziebart2008maximum,hadfield2016cooperative,ramachandran2007bayesian,boularias2011relative,levine2011nonlinear,choi2011inverse,ng2000algorithms}. 

Quantum computing has been a focal point of intense research and development efforts driven by its potential to revolutionize various fields from business, drug discovery and materials science to optimization and machine learning. From Shor's algorithm~\citep{shor1994algorithms} for integer factorization to Grover's algorithm~\citep{grover1996fast} for unstructured database search, quantum computing has showcased its prowess in solving computationally demanding tasks with unprecedented efficiency~\citep{rebentrost2014quantum,kerenidis2016quantum,lloyd2014quantum,kerenidis2019q,biamonte2017quantum,zhang2020recent,ambainis2007quantum,brassard2002quantum}. Seeing the power of quantum computing, a natural question is ``how can apprenticeship learning take advantage of quantum algorithms to enhance the capabilities of intelligent agents, enabling them to acquire expertise through observation and imitation in quantum-enhanced environments?". In this work, we study the problem of apprenticeship learning in the quantum setting by exploring how quantum algorithms can facilitate the acquisition of expert-level knowledge, paving the way for more capable and versatile quantum agents.

\subsection{Markov decision processes and apprenticeship learning}
A Markov Decision Process (MDP) can be represented as a 5-tuple $(\mathcal S, \mathcal A, R, P, \gamma)$, where $\mathcal S$ is a finite set of states with $|\mathcal S| = S$; $\mathcal A$ is a finite set of actions with $|\mathcal A| = A$; $R:\mathcal S\times \mathcal A\rightarrow \mathbb R$ is a reward function; $P = \{p(s'\vert s, a)\}_{s, a}$ is a set of transition probabilities, i.e. $p(s'\vert s, a)$ denotes the probability of transiting to state $s'$ after choosing action $a$ at state $s$; $\gamma\in [0, 1)$ is a discount factor. We let MDP$\backslash$R denote an MDP without a reward function, i.e. $(\mathcal S, \mathcal A, P, \gamma)$. 

A \emph{policy} $\pi$ is a mapping from states to a probability distribution over actions. The value of a policy on an MDP is called the \emph{value function}. For a given policy $\pi$, the value function $V^\pi:\mathcal S\rightarrow\left[0, \frac{1}{1 - \gamma}\right]$ is defined by $V^\pi(s) = \mathbb E\left[\displaystyle\sum_{t=0}^\infty \gamma^t R\left(s^{(t)}, a^{(t)}\right)\Big\vert \pi, s^{(0)} = s\right]$, where the expectation is taken over the random sequence of states drawn. It is known that any discounted MDP admits an optimal policy $\pi^*$ such that $V^{\pi^*}(s)\geq V^{\pi}(s)$ for all $\pi\in\Pi, s\in\mathcal S$, where $\Pi$ is the set of all policies~\citep{wang2021quantum}. On the other hand, for some small $\epsilon\in (0, 1)$, a policy $\tilde\pi$ is said to be \emph{$\epsilon$-optimal} if $V^{\pi^*}(s)\geq V^{\tilde\pi}(s) - \epsilon$ for all $s\in\mathcal S$. 

Often, when the state-action space is too large, the value function of an MDP can be approximated. A common class of approximators for value function approximation is the class of linear architectures~\citep{bradtke1996linear,melo2007q,osband2016generalization,azizzadenesheli2018efficient,yang2019sample}. Such an approximation corresponds to the use of linear parametric combination of basis functions, also known as \emph{feature vectors}~\citep{wei2021learning,wang2020reward,he2021logarithmic,jin2023provably}. 

In this paper, we are given feature vectors $\phi:\mathcal S\times \mathcal A\rightarrow [0, 1]^k$ over the state-action pairs. The collection $\left\{\phi(s, a)\right\}_{s, a}$ of feature vectors are stored in a feature matrix $\Phi\in\mathbb R^{SA\times k}$, where we assume query access to its entries. Furthermore, it is assumed that $\sup_{s\in\mathcal S, a\in\mathcal A}\left\Vert\phi (s, a)\right\Vert_2\leq 1$. In this work, we make the assumption that the reward function is linear, i.e. the ``true" reward function is given by $R^*(s, a) = w^*\cdot \phi(s, a)$, where $w^*\in\mathbb R^k$. It is also assumed that for all $s\in\mathcal S, a\in\mathcal A$, $\vert R(s, a)\vert\leq 1$\footnote{The same assumption was made in Ref.~\cite{abbeel2004apprenticeship}}. In order for this assumption to hold, we assume that $\Vert w^*\Vert_1\leq 1$. By linearity of the reward function, we can express the value function as 
\begin{align*}
    V^\pi(s)
    & = \mathbb E\left[\displaystyle\sum_{t=0}^\infty \gamma^t R\left(s^{(t)}, a^{(t)}\right)\Big\vert s^{(0)} = s, \pi\right]\\
    & = \mathbb  E\left[\displaystyle\sum_{t=0}^\infty \gamma^t w\cdot \phi\left(s^{}(t)), a^{(t)}\right)\Big\vert s^{(0)} = s, \pi\right]\\
    & = w\cdot \mathbb  E\left[\displaystyle\sum_{t=0}^\infty \gamma^t \phi\left(s^{(t)}), a^{(t)}\right)\Big\vert s^{(0)} = s, \pi\right],
\end{align*}
where the expectation is taken over the random sequence of states drawn by first starting from some initial state distribution $\mathcal D$ and choosing actions according to $\pi$. A \emph{feature expectation} is defined as the expected discounted accumulated feature value vector, i.e. 
\begin{align*}
    \mu(\pi\vert s) = \mathbb E\left[\displaystyle\sum_{t=0}^\infty \gamma^t \phi\left(s^{(t)}, a^{(t)}\right)\Big\vert s^{(0)} = s, \pi\right]\in\mathbb R^k. 
\end{align*}
Using this definition, one can rewrite 
\begin{align}\label{eqn:rewrite}
    V^{\pi}(s) = w\cdot \mu(\pi\vert s).
\end{align}
From now on, we shall drop the initial state $s$ in $\mu(\pi\vert s)$ and use the notation $\mu(\pi)$ instead when the initial state is irrelevant in the context. 

In the apprenticeship learning setting, the algorithm (apprentice) is required to ``learn" in an MDP where the reward function is not explicitly given. The algorithm is allowed to ``observe" demonstrations by an \emph{expert} and tries to maximize a reward function that is expressible as a linear combination of feature vectors. In order to do so, we assume access to an expert's policy $\pi_E$. In particular, we assume the ability to observe trajectories generated by an expert starting from $s^{(0)}\sim \mathcal D$ and choosing action according to $\pi_E$. 
Given an MDP$\backslash$R, feature matrix $\Phi$ and expert's feature expectations $\mu_E = \mu(\pi_E)$, apprenticeship learning aims to find a policy whose performance is close to that of the expert’s, on the unknown reward function $R^* = \Phi  w^*$. To achieve this, the algorithm finds a new policy $\pi^{(i)}$ at every iteration $i\in\{0, \cdots , n\}$, evaluates $\mu\left(\pi^{(i)}\right)$ and stores $\mu_E - \mu\left(\pi^{(i)}\right)$ in a matrix $\Phi'\in\mathbb R^{(n+2)\times k}$, whose first row corresponds to $\mu_E$ and row $i\in\{2, \cdots, n+2\}$ corresponds to $\mu_E - \mu\left(\pi^{(i-2)}\right)$. This process repeats until a policy $\bar \pi$ such that $\left\Vert \mu_E - \mu(\bar \pi)\right\Vert_2\leq\epsilon$ for some small $\epsilon\in (0, 1)$ is found. Hence, for any $w\in\mathbb R^k$ with $\Vert w\Vert_1\leq 1$, we have 
\begin{align*}
    \left\vert \mathbb E\left[\displaystyle\sum_{t=0}^\infty \gamma^t\ R\left(s^{(t)}, a^{(t)}\right)\Big\vert \pi_E\right] - \mathbb E\left[\displaystyle\sum_{t=0}^\infty \gamma^t R\left(s^{(t)}, a^{(t)}\right)\Big\vert\tilde\pi\right]\right\vert
    & = \left\vert w^T \mu_E - w^T\mu(\tilde\pi)\right\vert \\
    & \leq \left\Vert w\right\Vert_2\left\Vert \mu_E - \mu(\tilde\pi)\right\Vert_2 \\
    & \leq \epsilon. 
\end{align*}
The approach described above is known as inverse reinforcement learning~\citep{abbeel2004apprenticeship, ng2000algorithms}. 

\subsection{Our work}
Our main contribution is a quantum algorithm for apprenticeship learning via inverse reinforcement learning. As an intermediate step, we give a classical algorithm for approximate apprenticeship learning.  Our algorithms are based on the framework of~\cite{abbeel2004apprenticeship}. In this framework, the algorithms find a policy that performs as well as, or better than the expert's, on the expert’s unknown reward function. Our observation is that the robust analogue of this framework can be implemented using existing subroutines, both in the classical and quantum settings. This motivates us to study classical approximate and quantum apprenticeship learning algorithms, their convergence guarantees and to what extent the latter algorithm outperforms the former in terms of the time complexity. 

We show that using subroutines such as multivariate Monte Carlo, a sampling-based algorithm for linear classification of vectors and a reinforcement learning algorithm that returns a nearly optimal policy, our algorithms find a policy whose feature expectation is close to the expert's feature expectation at up to some error in the $\ell_2$-norm. Both our algorithms converge after $O\left(\frac{k}{(1 - \gamma)^2 (\epsilon^2 - \epsilon_{\operatorname{RL}})}\log \frac{k}{(1 - \gamma)^2\epsilon^2}\right)$ iterations, where $k$ is the dimension of the feature vectors, $(1 - \gamma)$ is the effective time horizon, $\epsilon_{\operatorname{RL}}$ is the error from using, as a subroutine, an approximate reinforcement learning algorithm that outputs an $\epsilon_{\operatorname{RL}}$-optimal policy, and $\epsilon$ is the error of the algorithm. Furthermore, our quantum algorithm is quadratically faster than its classical counterpart in $k$ and the size of the action space $A$. However, its time complexity scales worse in terms of the error dependence $\epsilon$ and the effective time horizon $(1 - \gamma)$. This is due to the tuning of error in certain subroutines in order to achieve convergence. Our quantum speedup comes from techniques that rely on amplitude estimation~\citep{brassard2002quantum}. Moreover, we leverage a recent computational model~\citep{allcock2023constant} comprising of a quantum processing unit and a quantum memory device to allow efficient moving and addressing of qubits.

We summarize our results in Table~\ref{table2}. 

\begin{table}[h]
\caption{Summary of results. In this work, $k$ is the dimension of the feature vectors; $\epsilon, \epsilon_{\operatorname{RL}}\in (0, 1)$ are errors of the algorithm and of the reinforcement learning subroutine respectively; $S$ and $A$ are sizes of the state and action spaces respectively; $\gamma\in[0, 1)$ is the discount factor.}\label{table2}%
\begin{tabular}{@{}llll@{}}
\toprule
Algorithms & Per-iteration time complexity  \\
\midrule
Classical approximate apprenticeship learning (Algorithm~\ref{algo}) & $\tilde O\left(\frac{k + SA}{(1 - \gamma)^7\epsilon^6(\epsilon^2 - \epsilon_{\operatorname{RL}})}\right)$ \\ 
Quantum apprenticeship learning (Algorithm~\ref{QAL}) & $\tilde O\left(\frac{\sqrt k + S\sqrt A}{(1 - \gamma)^{16}\epsilon^{24}(\epsilon^2 - \epsilon_{\operatorname{RL}})^{0.5}}\right). $\\
\botrule
\end{tabular}
\end{table}

The rest of the paper is organized as follows: In Section~\ref{sec:preliminaries}, we introduce useful notations, the computational model and related work. Section~\ref{sec:AL} briefly describes a slightly modified version of the algorithm by~\citep{abbeel2004apprenticeship} (Algorithm~\ref{AL}) and its convergence analysis. Next, we introduce the necessary classical subroutines and our classical approximate algorithm for apprenticeship learning (Algorithm~\ref{algo}), together with its time complexity analysis in Section~\ref{sec:approximate_AL}. In Section~\ref{sec:QAL}, we introduce the necessary quantum subroutines, our quantum apprenticeship learning algorithm (Algorithm~\ref{QAL}) and its running time analysis. Both our classical approximate and quantum algorithms have the same convergence guarantees, which follows from Algorithm~\ref{AL}. Lastly, we conclude our work in Section~\ref{sec:conclusion}. 

\section{Preliminaries}\label{sec:preliminaries}
\subsection{Notations}
For a positive integer $k\in\mathbb Z_+$, let $[k]$ denote the set $\{1, \cdots, k\}$. For a vector $v\in\mathbb R^k$, we use $v(i)$ to denote the $i$-th entry of $v$. The $\ell_2$- and $\ell_1$-norms of a vector $v\in\mathbb R^k$ are defined as $\left\Vert v\right\Vert_2\coloneq \sqrt{\displaystyle\sum_{i=1}^k \left(v(i)\right)^2}$ and $\Vert v\Vert_1\coloneqq \displaystyle\sum_{i=1}^k\left\vert v(i)\right\vert$ respectively.  For a matrix $M\in\mathbb R^{n\times k}$, we use $M(i)$ to denote the $i$-th row of $M$ and use $M(i, j)$ to denote the $(i, j)$-th entry of $M$. We use $\bar 0$ to denote the all zeroes vector and use $\ket{\bar 0}$ to denote $\ket{0}\otimes\cdots\otimes\ket{0}$ where the number of qubits is clear from the context. We use $\mathbf e$ to denote the all ones vector, where the dimension is clear from the context. We use $\tilde O(\cdot)$ to hide  the polylog factor, i.e. $\tilde O(f(n)) = O(f(n)\operatorname{polylog}f(n))$. 

\subsection{Computational model}
Our classical computational model is a classical random access machine. The input to the apprenticeship learning problem is a feature matrix $\Phi\in [0, 1]^{SA\times k}$ which is stored in a classical-readable read-only memory (ROM). Reading any entry of $\Phi$ takes constant time. The classical computer can write bits to a classical writable memory that stores the matrix $\Phi'\in\mathbb R^{(n+2)\times k}$, that is initialized to the all zeroes matrix at the beginning of the algorithm. As the algorithm proceeds, the first row of the matrix stores the vector $\hat\mu_E$ that estimates $\mu_E$. The subsequent rows of the matrix are updated with $\Phi'(i+1) = \hat\mu_E - \mu^{(i)}$ in every iteration $i\in\{0\}\cup[n]$. 

Our quantum computational model is a Quantum Processing Unit (QPU) and Quantum Memory Device (QMD)~\citep{allcock2023constant}. This computational model generalizes the Quantum Random Access Memory\footnote{A QRAM is a quantum analogue of a classical Random Access Memory that stores classical or quantum data, which allows queries to be performed in superposition.} (QRAM)~\citep{giovannetti2008architectures,Giovannetti2008} and the Quantum Random Access Gate (QRAG)\footnote{While a QRAM can be thought of as a “read-only” memory, a QRAG can be seen as a “read-write” memory since qubits are swapped from memory register into the target register, acted on, and then swapped back}.~\citep{ambainis2007quantum,aaronson2019quantum,buhrman2021limits,guo2022implementing}. 
In particular, it allows the following operations
\begin{align*}
    & \ket{i}\ket{B}\ket{x}\ket{\bar 0}\rightarrow \ket{i}\ket{B}\ket{x\oplus B_i}\ket{\bar 0}, \quad \forall i\in[n], B\in\{0, 1\}^{n\times k}, x\in\{0, 1\}^k\\
    & \ket{i}\ket{B}\ket{x}\ket{\bar 0}\rightarrow \ket{i}\otimes_{j=1}^{i-1}\ket{B_j}\ket{x}\otimes_{j=i+1}^{n}\ket{B_j}
    \ket{B_i}\ket{\bar 0}, \quad \forall i\in[n], B\in\{0, 1\}^{n\times k}, x\in\{0, 1\}^k.
\end{align*}
Here, $B$ is a $n\times k$ matrix (with $B_1, \ldots, B_n$ denoting its rows) stored in a quantum memory. The first operation reads a row of $B$, by XORing it with the current values in the register to which it is read. The second operation writes a row by swapping in values from an extra register.

We assume that the feature matrix $\Phi$ is stored in a QMD of $SA$ memory registers of size $k$. Specifically, the quantum computer has access to a feature matrix oracle $\mathcal O_\Phi$ which performs, for all $s\in\mathcal S, a\in\mathcal A$, the following mapping:
\begin{align*}
    \mathcal O_\Phi: \ket{s}\ket{a}\ket{\bar 0}\rightarrow \ket{s}\ket{a}\ket{\phi(s, a)}, 
\end{align*}
where $\phi(s, a)\in [0, 1]^k$ denotes the $(s, a)$-th row of $\Phi$. The second register is assumed to contain sufficient qubits to ensure the accuracy of subsequent computations, in analogy to the sufficient number of bits that a classical algorithm needs to run correctly. 

Our quantum algorithm shall commonly build $\mathsf{KP}$-trees of vectors $\hat\mu_E$\footnote{$\hat \mu_E$ is an approximate of $\mu_E$, more details in Section~\ref{sec:AL}.} and $\hat\mu_E - \mu^{(i)}$ for all $i\in \{0\}\cup [n]$. These $\mathsf{KP}$-trees are collectively called $\mathsf{KP}_{\Phi'}$. A $\mathsf{KP}$-tree is a classical data structure introduced by~\cite{kerenidis2016quantum,  prakash2014quantum} to store vectors or matrices and facilitates efficient quantum state preparation.  
\begin{fact}[$\mathsf{KP}$-tree~\citep{kerenidis2016quantum, prakash2014quantum}]\label{fact:KP-tree}
    Let $M\in\mathbb R^{m\times n}$ be a matrix with $w \in \mathbb N$ non-zero entries. There is a data structure of size $O(w\operatorname{poly}\log{mn})$ that stores each input $(i,j, M(i, j))$ in time $O(\operatorname{poly}\log{mn})$. Furthermore, finding $\Vert M\Vert_F^2$ and $\Vert M(i)\Vert_2$ takes $O(1)$ and $O(\log n)$ time respectively. 
\end{fact} 

The quantum computer can update the $\mathsf{KP}$-trees via the QMD. Namely, at iteration $i$, it writes a vector $\Phi'(i+1) = \hat\mu_E - \mu\left(\pi^{(i)}\right)\in\mathbb R^k$ into the memory, which allows the quantum computer to invoke the oracle $\mathcal O_{\Phi'}$  that performs the mapping
\begin{align*}
        \mathcal O_{\Phi'}\ket{j}\ket{\bar 0}\rightarrow\ket{j}\ket{\hat\mu_E - \mu\left(\pi^{(j)}\right)}, \quad j\in \{0\}\cup [n]. 
\end{align*}
We refer to the ability to invoke oracles $\mathcal O_{\Phi}$ and $\mathcal O_{\Phi'}$ as having quantum access to $\Phi $ and $\Phi'$.

For the MDP, we consider the classical generative model studied in~\cite{wang2021quantum, kearns1998finite, kearns2002sparse, kakade2003sample}, which 
allows us to collect $N$ i.i.d. samples 
\begin{align*}
    s^i_{s, a}\sim p(\cdot\vert s, a), \quad i\in[N]
\end{align*}
given a state-action pair $(s, a)$. This enables the construction of an empirical transition kernel $\hat P = \{\hat p(s'\vert s, a)\}_{s, a}$ satisfying 
\begin{align*}
    \hat p(s'\vert s, a) = \frac{1}{N} \sum_{i=1}^N\mathbbm 1_{\{s^i_{s, a} = s'\}}, \quad \forall s'\in\mathcal S,
\end{align*}
where $\mathbbm 1_{\{\cdot\}}$ is the indicator function. This results in an empirical MDP with parameters $(\mathcal S, \mathcal A, R, \hat P, \gamma)$. We assume that a single call to the generative model takes constant time. 

While trajectories can be observed classically, the same cannot be done in the quantum setting. Following the quantum-accessible environments studied by~\cite{wang2021quantum, wiedemann2022quantum, jerbi2022quantum, zhong2023provably}, we require quantum oracles to access the state-action trajectories and parameters of an MDP at each time step $t\in\{0, \cdots, H\}$, for some $H\in\mathbb Z_+$. 
\begin{enumerate}
    \item A transition probability oracle $\mathcal O_{P}$ that returns a superposition over the next states according to the transition probability
    \begin{align*}
        \mathcal O_{P}:\ket{s^{(t)}}\ket{a^{(t)}}\ket{0}\rightarrow \ket{s^{(t)}}\ket{a^{(t)}}\otimes \displaystyle\sum_{s^{(t+1)}\in\mathcal S}\sqrt{p\left(s^{(t+1)}\vert s^{(t)}, a^{(t)}\right)}\ket{s^{(t+1)}}
\end{align*}
    Note that if one performs a measurement on all three registers after querying the oracle, the result will be equivalent to drawing one sample in the classical generative model.
    \item A reward oracle $\mathcal O_R$ such that given access to the feature matrix $\Phi\in[0, 1]^{SA\times k}$ and a vector $w\in\mathbb R^k$, performs the following mapping:
    \begin{align*}
        \mathcal O_R:\ket{s^{(t)}}\ket{a^{(t)}}\ket{\bar 0}\rightarrow \ket{s^{(t)}}\ket{a^{(t)}}\ket{w^T\Phi\left(s^{(t)}, a^{(t)}\right)}. 
    \end{align*}
\end{enumerate}
We assume that a single query to the transition probability and reward oracles takes $O(1)$ time.  

We refer to the running time of a classical/quantum computation as the number of elementary gates performed, plus the number of calls to the classical/quantum memory device, excluding the gates that are used inside the oracles for the quantum accessible environments. We assume a classical arithmetic model, which allows issues arising from the fixed-point representation of real numbers to be ignored. In this model, elementary arithmetic operations take constant time.  In the quantum setting, we assume a quantum circuit model. Every quantum gate in the circuit represents an elementary operation, and the application of every quantum gate takes constant time. The time complexity of a given unitary operator $U$ is the minimum number of elementary quantum gates needed to prepare $U$. In addition, we assume a quantum arithmetic model, which is equivalent to the classical model, i.e. arithmetic operations take constant time.

\subsection{Related work}
In the apprenticeship learning setting,~\cite{abbeel2004apprenticeship} studied the framework of learning in an MDP where the reward function is not explicitly given. Given access to demonstrations by an expert that tries to maximize a linear reward function, they gave an algorithm using an inverse reinforcement learning approach, that outputs an approximate of the true reward function and finds a policy that performs at least as well as, or better than that of the expert on the unknown reward function. The same authors also considered reinforcement learning in systems with unknown dynamics~\citep{abbeel2005exploration}. They showed that given initial demonstrations by an expert, no explicit exploration is necessary, and a near-optimal performance can be obtained, simply by repeatedly executing exploitation policies that try to maximize rewards. In the semi-unsupervised apprenticeship learning setting, many sample trajectories are observed but only a few of them are labeled as the experts’ trajectories. \cite{valko2013semi} defined an extension of the max-margin inverse reinforcement learning by~\cite{abbeel2004apprenticeship}, and studied the conditions under which the unlabeled trajectories can be helpful in learning good performing policies. They showed empirically that their algorithm outputs a better policy in fewer iterations than the algorithm by~\cite{abbeel2004apprenticeship} that does not take the unlabeled trajectories into account.

In the field of inverse reinforcement learning,~\cite{ng2000algorithms} first characterized the set of all reward functions for which a given policy is optimal. They then derived three algorithms: two of which are in the setting where the complete policy is known, and the third algorithm is in the setting where the policy is known through a finite set of observed trajectories. As a result of removing degeneracy\footnote{The existence of a large set of reward functions for which the observed policy is optimal}, inverse reinforcement learning is formulated into an efficiently solvable linear program.~\cite{fu2017learning} proposed a practical and scalable inverse reinforcement learning algorithm based on an adversarial reward learning formulation. In particular, they showed that their algorithm is capable of recovering reward functions that are robust to changes in dynamics, enabling policies to be learned  despite significant variation in the environment seen during training. By posing inverse reinforcement learning as a Bayesian learning task,~\cite{ramachandran2007bayesian}  showed that improved solutions can be obtained. In particular, they provided a theoretical framework and tractable algorithms for Bayesian inverse reinforcement learning. Their numerical results show that the solutions output by their algorithm are more informative in terms of the reward structure as compared to that of~\cite{ng2000algorithms}. More related work on inverse reinforcement learning can be found in survey papers by~\cite{arora2021survey, zhifei2012survey, adams2022survey, zhifei2012review}. 

Seeing the wide application of reinforcement learning in diverse disciplines from healthcare, robotics and autonomous, communication and networking, natural language processing and computer vision~\citep{yu2021reinforcement,kormushev2013reinforcement,kober2013reinforcement,luong2019applications,le2022deep}, theoretical research on reinforcement learning algorithms seek to design faster and more resource efficient learning methods that can generalize across different domains. These efforts aim to bridge the gap between theoretical insights and practical applications. Policy iteration and value iteration are among the fundamental techniques in reinforcement learning, facilitating the convergence towards an optimal policy~\citep{lagoudakis2003least,bertsekas2011approximate,puterman1978modified, pineau2003point,li2020breaking,silver2014deterministic}. 

In the quantum setting, reinforcement learning algorithms gain a speedup over their classical counterparts, thanks to tools such as amplitude amplification and estimation~\citep{brassard2002quantum}, quantum mean estimation~\citep{montanaro2015quantum} and block-encoding techniques~\citep{low2019hamiltonian, chakraborty2018power}. Such algorithms include~\citep{wiedemann2022quantum,wang2021quantum,jerbi2022quantum,cherrat2023quantum}. Further related work on quantum reinforcement learning can be found in survey papers by~\cite{dong2008quantum,dunjko2017advances,meyer2022survey}. For the task of linear classification,~\cite{rebentrost2014quantum} gave a quantum algorithm based on the Harrow-Hassidim- Lloyd (HHL) algorithm~\citep{harrow2009quantum} and quantum linear algebra techniques. In the case of general data sets,~\cite{li2019sublinear} gave a quantum algorithm that runs in time $\tilde O\left(\frac{\sqrt n}{\epsilon^4} + \frac{\sqrt d}{\epsilon^8}\right)$, improving over the classical running time of $O\left(\frac{n+d}{\epsilon^2}\right)$ by~\cite{clarkson2012sublinear}. The complexity of the quantum Support Vector Machine (SVM) has been studied by~\cite{gentinetta2024complexity}.

\section{Apprenticeship learning via inverse reinforcement learning}\label{sec:AL}
\cite{abbeel2004apprenticeship} gave an algorithm for (apprenticeship) learning an unknown reward function using inverse reinforcement learning. First, an estimate of the expert’s feature expectations $\mu_E \coloneqq \mu(\pi_E)$ is obtained and stored in the first row of $\Phi'$. More specifically, given a set of $m$ trajectories $\{s^{(0)}_i, a^{(0)}_i, s^{(1)}_i, a^{(1)}_i, \cdots \}_{i=1}^m$ 
generated by the expert, denote the empirical estimate for $\mu_E$ as
\begin{align}\label{eqn:expert_feature_expectation}
    \hat\mu_E \coloneqq \frac{1}{m}\displaystyle\sum_{i=1}^m \displaystyle\sum_{t=0}^\infty \gamma^t \phi\left(s^{(t)}_i, a^{(t)}_i\right)
\end{align}
and $\Phi'(1) = \hat\mu_E$. The algorithm then proceeds as follows: the idea is to first randomly pick a policy $\pi^{(0)}$ and compute (or approximate) its feature expectation $\mu^{((0)} \coloneqq \mu\left(\pi^{(0)}\right)$. At any iteration $j\in\mathbb Z_+$, $w^{(j)}$ is the vector of $\ell_2$-norm at most 1, that maximizes the minimum of inner product with  $\mu_E - \mu\left(\pi^{(i)}\right)$ for $i \in \{0, \cdots, (j-1)\}$. The (presumably optimal) policy for the next iteration is then generated using any reinforcement learning algorithm augmented with $R = \Phi\cdot w^{(j)}$ as the reward function. This process is repeated until a $w^{(j')}$ satisfying $\left\Vert w^{(j')^T}\left(\displaystyle\min_{i = \{0, \cdots, (j'-1)\}}\mu_E - \mu\left(\pi^{(i)}\right)\right)\right\Vert_2\leq \epsilon$ for some $\epsilon\in (0, 1)$ is obtained. 

We slightly modify the algorithm of~\cite{abbeel2004apprenticeship} in a way such that the policies generated are $\epsilon_{\operatorname{RL}}$-optimal for some $\epsilon_{\operatorname{RL}}\in (0, 1)$. The modified algorithm is summarized in Algorithm~\ref{AL}. 

\begin{breakablealgorithm}
    \caption{Apprenticeship learning algorithm}
    \label{AL}
    \begin{algorithmic}[1]
    \Require Initialize policy $\tilde \pi^{(0)}$, errors $\epsilon, \epsilon_{\operatorname{RL}}\in (0, 1)$ such that $\epsilon\geq\sqrt{\epsilon_{\operatorname{RL}}}$. 
    \Ensure $\left\{\tilde \pi^{(i)}: i = 0, \cdots, n\right\}$. 
    \State Compute (or approximate via Mote Carlo) $\mu^{(0)} \coloneqq \mu\left(\tilde \pi^{(0)}\right)$. 
    \State Set $i = 1$.
    \State \label{line:0SVM} Compute $t^{(i)} = \displaystyle\max_{w:\Vert w\Vert_2\leq 1}\min_{j\in \{0, \cdots, (i-1)\}}w^T\left(\mu_E - \mu^{(j)}\right)$ and let $w^{(i)}$ be the corresponding argument maximum. 
    \State If $t^{(i)}\leq \epsilon$, then terminate and set $n = i$. 
    \State Compute an $\epsilon_{\operatorname{RL}}$-optimal policy $\tilde \pi^{(i)}$ for the MDP using rewards $R = \Phi w^{(i)}$ 
    \State Compute (or estimate)  $\mu^{(i)} \coloneqq \mu\left(\tilde  \pi^{(i)}\right)$. 
    \State Set $i = i+1$ and go to Step 3. 
\end{algorithmic}
\end{breakablealgorithm}

In Figure~\ref{fig:1}, we give a workflow diagram to illustrate the main steps of the apprenticeship learning algorithm.  

\begin{figure}[H]\label{fig:1}
\begin{center}
\begin{tikzpicture}[node distance=2cm]
\node (input) {Input: Arbitrarily chosen random policy};
\node (ME) [process, below of=input] {\textbf{Mean estimation}: \\ Estimate the mean of feature vector};
\node (SVM) [process, below of=ME] {\textbf{Support vector machine solver}:\\ Find the weight vector};
\node (stop) [right of=SVM, xshift=5cm] {Terminate};
\node (RL) [process, below of=SVM] {\textbf{Reinforcement learning algorithm}: \\ Obtain an $\epsilon_{\operatorname{RL}}$-optimal policy};
\draw [->] (input) to node[midway, right] {}(ME);
\draw  [->] (ME) to node[midway, right] {} (SVM);
\draw [->] (SVM) to node[midway, right] {if $t^{(i)} > \epsilon$} (RL);
\draw [->] (SVM) to node[midway, above] {if $t^{(i)} \leq \epsilon$} (stop);
\coordinate (leftRL) at ([xshift=-1.5cm] RL.west);
\coordinate (leftME) at ([xshift=-1.5cm] ME.west);
\draw[->] (RL.west) -- (leftRL) -- (leftME) -- (ME.west);
\end{tikzpicture}
\caption{An illustration of the apprenticeship learning algorithm.}
\end{center}
\end{figure}

Before we delve into the convergence analysis, we define some notations that we will be using in the remaining part of this subsection. Given a set of policies $\tilde \Pi$, we define $\tilde M = \operatorname{Co}\{\mu(\tilde\pi): \tilde \pi\in\tilde \Pi\}$ to be the convex hull of the set of feature expectations attained by policies $\tilde \pi\in\tilde \Pi$. Given any vector of feature expectations $\tilde\mu\in \tilde M$, there exists a set of policies $\tilde \pi_1, \cdots, \tilde \pi_n\in\tilde \Pi$ and weights $\{\lambda_i\}_{i=1}^n$ such that $\lambda_i \geq 0$ for all $i\in [n]$, $\sum_{i=1}^n \lambda_i = 1$ and $\tilde\mu = \sum_{i=1}^n \lambda_i \mu(\tilde \pi_i)$. Therefore, given any point $\tilde\mu\in \tilde  M$, we can obtain a new policy whose feature expectation is exactly $\tilde\mu$ by taking a convex combination of policies in $\tilde \Pi$. We also define $\tilde M^{(i)} = \operatorname{Co}\{\mu(\tilde \pi(j)): j = 0,\cdots,i\}$ to be the convex hull of the set of feature expectations of policies found after iterations $0,\cdots, i$ of the algorithm. Furthermore, define $\bar\mu^{(i)}\coloneqq \argmin_{\mu\in \operatorname{Co}\{\bar\mu^{(i)}, \mu^{(i+1)}\}} \left\Vert \hat\mu_E - \mu\right\Vert_2$  and hence $\bar\mu^{(i)} \in \tilde M$. 

The following lemma establishes improvement in a single iteration of Algorithm~\ref{AL}.
\begin{lemma}[Per-iteration improvement of Algorithm~\ref{AL}]\label{lem:lemma_3}
    Let there be given an $MDP\backslash R$, feature vectors $\phi:\mathcal S \rightarrow [0, 1]^k$ and a set of policies $\tilde\Pi$, $\bar\mu^{(i)}\in \tilde M$. Let $\epsilon
    _{\operatorname{RL}} \in (0, 1)$ be such that $\left\Vert \hat\mu_E\right\Vert_2^2\geq 2\epsilon_{\operatorname{RL}}$ and let $\tilde\pi^{(i)}$ be the $\epsilon_{\operatorname{RL}}$-optimal policy for the $MDP\backslash R$ augmented with reward $R(s) = \left(\hat\mu_E - \bar\mu^{(i)}\right)\cdot \phi(s)$, i.e. $(\mu_E - \bar\mu^{(i)})\cdot \mu\left(\tilde pi^{i+1}\right) \geq  \argmax_\pi (\mu_E - \bar\mu^{(i)})\cdot \mu(\pi) - \epsilon_{\operatorname{RL}}$. Finally, let $\tilde \mu^{(i+1)} = \frac{\left(\hat\mu_E - \bar\mu^{(i)}\right)\cdot \left(\mu^{(i+1)} - \bar\mu^{(i)}\right)- \epsilon_{\operatorname{RL}}}{\left\Vert \mu^{(i+1)} - \bar\mu^{(i)}\right\Vert_2^2}\left(\mu^{(i+1)} - \bar\mu^{(i)}\right) + \bar\mu^{(i)}$, i.e. the projection of $\hat\mu_E$ onto the line through $\mu^{(i+1)} = \mu\left(\tilde\pi^{(i+1)}\right)$, $\bar\mu^{(i)}$. Then, 
    \begin{align*}
        \frac{\Vert \hat\mu_E - \tilde\mu^{(i+1)\Vert_2}}{\Vert \hat\mu_E - \bar\mu^{(i)}\Vert_2}\leq  \frac{\sqrt k + (1 - \gamma)\sqrt{\epsilon_{\operatorname{RL}}/2}}{\sqrt{k + (1 - \gamma)^2 \left(\left\Vert \hat\mu_E  - \bar\mu^{(i)}\right\Vert^2_2 - \epsilon_{\operatorname{RL}}\right)}}
    \end{align*}
    and the point $\tilde\mu^{(i+1)}$ is a convex combination of $\bar\mu^{(i)}$ and $\mu^{(i+1)}$.
\end{lemma}
\begin{proof}
For simplicity of notation, we let the origin of our coordinate system coincide with $\bar\mu^{(i)}$. Then, 
\begin{align}
    & \nonumber \frac{\left(\tilde\mu^{(i+1)} - \hat\mu_E\right)\cdot \left(\tilde\mu^{(i+1)} - \hat\mu_E\right)}{\hat\mu_E - \hat\mu_E} \\
    & \label{eq:1}  = \frac{\mu^{(i+1)} \cdot \mu^{(i+1)} +  \frac{ \epsilon^2_{\operatorname{RL}} -\left(\mu^{(i+1)}\cdot \hat\mu_E \right)^2 }{\hat\mu_E \cdot \hat\mu_E}}{\mu^{i+1} \cdot \mu^{(i+1)}} \\
    & \label{eq:2} \leq \frac{\mu^{(i+1)} \cdot \mu^{(i+1)} +  \epsilon_{\operatorname{RL}}/2 - 2\left(\mu^{(i+1)} \cdot \hat\mu_E\right) +  \hat\mu_E \cdot \hat\mu_E}{\mu^{(i+1)} \cdot \mu^{(i+1)}} \\
    & = \frac{\left(\mu^{(i+1)} - \hat\mu_E\right) \cdot \left(\mu^{(i+1)} - \hat\mu_E\right) + \epsilon_{\operatorname{RL}}/2}{\left(\mu^{(i+1)} - \hat\mu_E\right) \cdot \left(\mu^{(i+1)} - \hat\mu_E\right) +  \hat\mu_E \cdot \hat\mu_E + 2\left(\mu^{(i+1)} - \hat\mu_E\right)\cdot \hat\mu_E} \\
    & \label{eq:4} \leq \frac{\left(\mu^{(i+1)} - \hat\mu_E\right) \cdot \left(\mu^{(i+1)} - \hat\mu_E\right) + \epsilon_{\operatorname{RL}}/2}{\left(\mu^{(i+1)} - \hat\mu_E\right) \cdot \left(\mu^{(i+1)} - \hat\mu_E\right) + \hat\mu_E \cdot \hat\mu_E - \epsilon_{\operatorname{RL}}} \\
    & \leq \frac{k /(1 - \gamma)^2 + \epsilon_{\operatorname{RL}}/2}{k/(1 - \gamma)^2 + \hat\mu_E \cdot \hat\mu_E - \epsilon_{\operatorname{RL}}}
\end{align} 
where we use in the order: 
\begin{enumerate}
    \item The definition of $\tilde \mu^{(i+1)} = \frac{\hat\mu_E \cdot \mu^{(i+1)} - \epsilon_{\operatorname{RL}}}{\mu^{(i+1)}\cdot \mu^{(i+1)}}\mu^{(i+1)}$, which gives for the numerator 
    \begin{align*}
        & \left(\tilde\mu^{(i+1)} - \hat\mu_E\right) \cdot \left(\tilde\mu^{(i+1)} - \hat\mu_E\right)\\
        & = \left(\frac{\hat\mu_E \cdot \mu^{(i+1)} -\epsilon_{\operatorname{RL}}}{\mu^{(i+1)}\cdot \mu^{(i+1)}}\mu^{(i+1)} - \hat\mu_E \right) \cdot \left(\frac{\hat\mu_E \cdot \mu^{(i+1)} - \epsilon_{\operatorname{RL}}}{\mu^{(i+1)}\cdot \mu^{(i+1)}}\mu^{(i+1)} - \hat\mu_E \right)\\
        & = \frac{\left(\hat\mu_E \cdot \mu^{(i+1)} -\epsilon_{\operatorname{RL}}\right)^2}{\left(\mu^{(i+1)}\cdot \mu^{(i+1)}\right)^2}\mu^{(i+1)}\cdot \mu^{(I+1)} - 2\frac{\hat\mu_E \cdot \mu^{(i+1)} -\epsilon_{\operatorname{RL}}}{\mu^{(i+1)}\cdot \mu^{(i+1)}}\mu^{(i+1)}\cdot \hat\mu_E + \hat\mu_E \cdot \hat\mu_E \\
        & = \frac{\left(\hat\mu_E \cdot \mu^{(i+1)}\right)^2 - 2\epsilon_{\operatorname{RL}} \hat\mu_E \cdot \mu^{(i+1)} +\epsilon^2_{\operatorname{RL}}}{\left(\mu^{(i+1)}\cdot \mu^{(i+1)}\right)^2}\mu^{(i+1)}\cdot \mu^{(I+1)} - 2\frac{\hat\mu_E \cdot \mu^{(i+1)} -\epsilon_{\operatorname{RL}}}{\mu^{(i+1)}\cdot \mu^{(i+1)}}\mu^{(i+1)}\cdot \hat\mu_E \\
        & + \hat\mu_E \cdot \hat\mu_E \\
        & = \frac{\epsilon^2_{\operatorname{RL}}}
        {\mu^{(i+1)}\cdot \mu^{(i+1)}} - \frac{\left(\mu^{(i+1)}\cdot \hat\mu_E\right)^2}{\mu^{(i+1)}\cdot \mu^{(i+1)}} + \hat\mu_E \cdot \hat\mu_E \\
    \end{align*}
    Using this expression as the numerator, and multiplying both numerator and denominator by $\frac{\mu^{(i+1)} \cdot \mu^{(i+1)}}{\hat\mu_E \cdot \hat\mu_E}$ gives Equation~\ref{eq:1}. 
    \item To get Equation~\ref{eq:2}, note that 
    \begin{align*}
        & \left(\mu^{(i+1)} \cdot \hat\mu_E - \hat\mu_E \cdot \hat\mu_E\right)^2 \geq 0\\
        & \left(\mu^{(i+1)} \cdot \hat\mu_E\right)^2 - 2\left(\mu^{(i+1)} \cdot \hat\mu_E\right)\left(\hat\mu_E \cdot \hat\mu_E \right) + \left(\hat\mu_E \cdot \hat\mu_E\right)^2 \geq 0\\
        & - \left(\mu^{(i+1)} \cdot \hat\mu_E\right)^2 \leq - 2 \left(\mu^{(i+1)} \cdot \hat\mu_E\right)\left(\hat\mu_E \cdot \hat\mu_E \right) + \left(\hat\mu_E \cdot \hat\mu_E\right)^2 \\
        & \frac{- \left(\mu^{(i+1)} \cdot \hat\mu_E\right)^2}{\hat\mu_E \cdot \hat\mu_E} \leq -2\left(\mu^{(i+1)} \cdot \hat\mu_E\right) + \hat\mu_E \cdot \hat\mu_E \\
    \end{align*}
    and $\frac{\epsilon^2_{\operatorname{RL}}}{\hat\mu_E \cdot \hat\mu_E } \leq \frac{
    \epsilon^2_{\operatorname{RL}}}{2\epsilon_{\operatorname{RL}}} = \frac{\epsilon_{\operatorname{RL}}}{2}$
    by the assumption that $\hat\mu_E \cdot \hat\mu_E \geq 2\epsilon_{\operatorname{RL}}$. 
    \item For the numerator, use the fact that $\left(\mu^{(i+1)} - \hat\mu_E\right) \cdot \left(\mu^{(ii+1)} - \hat\mu_E\right) = \mu^{(i+1)} \cdot \mu^{(i+1)} - 2 \mu^{(i+1)} \cdot \hat\mu_E + \hat\mu_E \cdot \hat\mu_E$. For the denominator, we rewrite it as follows: 
    \begin{align*}
        \mu^{(i+1)} \cdot \mu^{(i+1)}
        & = \left(\mu^{(i+1)} - \hat\mu_E + \hat\mu_E\right) \cdot \left(\mu^{(i+1)} - \hat\mu_E + \hat\mu_E\right) \\
        & = \left(\mu^{(i+1)} - \hat\mu_E \right) \cdot \left(\mu^{(i+1)} - \hat\mu_E \right) + \left(\hat\mu_E \cdot \hat\mu_E\right) + 2\left(\mu^{)(i+1)} - \hat\mu_E\right)\cdot \hat\mu_E. 
    \end{align*}
    \item Since $\tilde\pi^{(i+1)}$ is ann $\epsilon_{\operatorname{RL}}$-optimal policy, we have 
    \begin{align*}
        \hat\mu_E \cdot \mu^{(i+1)} \geq \argmax_\pi  \hat\mu_E \cdot \mu(\pi) - \epsilon_{\operatorname{RL}} \geq \hat\mu_E \cdot \hat\mu_E - \epsilon_{\operatorname{RL}}, 
    \end{align*}
    which implies that $2\left(\mu^{(i+1)}- \hat\mu_E \right)\cdot \hat\mu_E \geq - \epsilon_{\operatorname{RL}}$, which in turn implies Equation~\ref{eq:4}. 
    \item Recall that $\phi\in[0, 1]^k$ and hence all $\mu\in[0, \frac{1}{1 - \gamma}]^k = \tilde M$. All points considered lie in $\tilde M$ so their $\ell_2$-norms are bounded by $\frac{\sqrt k}{1 - \gamma}$. 
\end{enumerate}
Now, we proof that $\tilde \mu^{(i+1)} = \lambda \mu^{(i+1)} + (1 - \lambda) \bar\mu^{(i)}$, for some $\lambda \in [0, 1]$. It is easy to see that, from the definition of $\tilde\mu^{(i+1)}$, by setting $\lambda = \frac{\hat\mu_E \cdot \mu^{(i+1)} - \epsilon_{\operatorname{RL}}}{\mu^{(i+1)}\cdot \mu^{(i+1)}}$, we have $\tilde \mu^{(i+1)} = \lambda \mu^{(i+1)} + (1 - \lambda) \bar\mu^{(i)}$\footnote{Keeping in mind that we use $\bar\mu^{(i)}$ as the original for notation simplicity.}. Since $\tilde\pi^{(i+1)}$ satisfies 
    \begin{align*}
        \hat\mu_E \cdot \mu\left(\tilde \pi^{(i+1)} \right) = \hat\mu_E \cdot \mu^{(i+1)} \geq \argmax_\pi  \hat\mu_E \cdot \mu(\pi) - \epsilon_{\operatorname{RL}} \geq \hat\mu_E \cdot \hat\mu_E - \epsilon_{\operatorname{RL}}, 
    \end{align*}
    we have 
    \begin{align*}
        \hat\mu_E \cdot \mu^{(i+1)} - \epsilon_{\operatorname{RL}} \geq \hat\mu_E \cdot \hat\mu_E - 2\epsilon_{\operatorname{RL}} \geq 0
    \end{align*}
    implying that $\lambda \geq 0$. On the other hand, observe that 
    \begin{align*}
        & - \left(\mu^{(i+1)} - \hat\mu_E\right) \left(\mu^{(i+1)} - \hat\mu_E\right) \leq 0\\
        & - \mu^{(i+1)} \cdot \mu^{(i+1)} + 2 \mu^{(i+1)} \cdot \hat\mu_E - \hat\mu_E \cdot \hat\mu_E \leq 0 \\
        & 2 \mu^{(i+1)} \cdot \hat\mu_E -  \hat\mu_E \cdot \hat\mu_E \leq \mu^{(i+1)} \cdot \mu^{(i+1)}. 
    \end{align*}
    Subtracting both sides by  $\mu^{(i+1)}\cdot\hat\mu_E$ gives
    \begin{align*}
    \mu^{(i+1)}\cdot\mu^{(i+1)} - \mu^{(i+1)}\cdot\hat\mu_E 
        & \geq 2\mu^{(i+1)}\cdot \hat\mu_E - \hat\mu_E \cdot \hat\mu_E - \mu^{(i+1)}\cdot\hat\mu_E \\
        & = \mu^{(i+1)}\cdot\hat\mu_E  - \hat\mu_E \cdot \hat\mu_E \\
        & \geq -\epsilon_{\operatorname{RL}}. 
    \end{align*}
    Dividing throughout by $\mu^{(i+1)}\cdot \mu^{(i+1)}$ and rearranging the terms gives $\lambda \leq 1$. 
\end{proof} 

In the theorem below, we prove the convergence guarantee of Algorithm~\ref{AL}. 
\begin{theorem}[Convergence guarantee of~Algorithm~\ref{AL}]\label{thm:iterations}
    Let an $MDP\backslash R$, features $\phi:\mathcal S \rightarrow [0, 1]^k$, and any $\epsilon_{\operatorname{RL}}, \epsilon\in (0, 1)$ such that $\epsilon \geq \sqrt{\epsilon_{\operatorname{RL}}}$ be given. Then Algorithm~\ref{AL} will terminate after at most
    \begin{align*}
        O\left(\frac{k}{(1 - \gamma)^2 (\epsilon^2 - \epsilon_{\operatorname{RL}})}\log \frac{k}{(1 - \gamma)^2\epsilon^2}\right)
    \end{align*}
    iterations. 
\end{theorem}
\begin{proof}
    Lemma~\ref{lem:lemma_3} gives a construction of $\tilde\mu^{(i+1)}\in\tilde M^{(i+1)}$ with a distance to $\hat\mu_E$ that is at most a factor given by the bound in the lemma statement. As long as $\bar\mu^{(i)}$ in the current iteration $i$ satisfies $\left\Vert \hat\mu_E - \bar\mu^{(i)}\right\Vert_2\geq \epsilon$, we have $\frac{\Vert \hat\mu_E - \tilde\mu^{(i+1)\Vert_2}}{\Vert \hat\mu_E - \bar\mu^{(i)}\Vert_2}\leq  \frac{\sqrt k + (1 - \gamma)\sqrt{\epsilon_{\operatorname{RL}}/2}}{\sqrt{k + (1 - \gamma)^2 \left(\epsilon^2 - \epsilon_{\operatorname{RL}}\right)}}$. 

    Let $t^{(i)} = \left\Vert\bar\mu^{(i)} - \hat\mu_E \right\Vert_2$. Define $\tilde\mu^{(i+1)}$ as in Lemma~\ref{lem:lemma_3} and observing that $\tilde\mu^{(i+1)}\in\tilde M^{(i+1)}$, then by definition,  $t^{(i+1)} \leq \left\Vert\tilde\mu^{(i+1)} - \hat\mu_E\right\Vert_2$. Since the maximum distance in $\tilde M$ is $\frac{\sqrt k}{1 - \gamma}$, we have 
    \begin{align*}
        t^{(i)} 
        & \leq \left(\frac{\sqrt k + (1 - \gamma)\sqrt{\epsilon_{\operatorname{RL}}/2}}{\sqrt{k + (1 - \gamma)^2 \left(\epsilon^2 - \epsilon_{\operatorname{RL}}\right)}}\right)^i \cdot t^{(0)}\\
        & \leq \left(\frac{\sqrt k + (1 - \gamma)\sqrt{\epsilon_{\operatorname{RL}}/2}}{\sqrt{k + (1 - \gamma)^2 \left(\epsilon^2 - \epsilon_{\operatorname{RL}}\right)}}\right)^i \cdot \frac{\sqrt k}{1 - \gamma}. 
    \end{align*}
    So, $t^{(i)} \leq \epsilon$ when 
    \begin{align*}
        \left(\frac{\sqrt k + (1 - \gamma)\sqrt{\epsilon_{\operatorname{RL}}/2}}{\sqrt{k + (1 - \gamma)^2 \left(\epsilon^2 - \epsilon_{\operatorname{RL}}\right)}}\right)^i \cdot \frac{\sqrt k}{1 - \gamma}. \leq\epsilon, 
    \end{align*}
    which is equivalent to 
    \begin{align}\label{eq:a/b}
        i \geq \log \left(\frac{\sqrt k}{(1 - \gamma)\epsilon}\right)/\log \left(\frac{\sqrt{k + (1 - \gamma)^2 (\epsilon^2 - \epsilon_{\operatorname{RL}})}}{\sqrt k + (1 - \gamma) \sqrt{\epsilon_{\operatorname{RL}}/2}}\right)
    \end{align}
    The denominator can be bounded by 
    \begin{align}
        \nonumber \log \left(\frac{\sqrt{k + (1 - \gamma)^2 (\epsilon^2 - \epsilon_{\operatorname{RL}})}}{\sqrt k + (1 - \gamma) \sqrt{\epsilon_{\operatorname{RL}}/2}}\right) 
        & \nonumber \leq \log \left(\frac{\sqrt{k + (1 - \gamma)^2 (\epsilon^2 - \epsilon_{\operatorname{RL}})}}{\sqrt k}\right)\\
        & \nonumber = \frac{1}{2} \log \left(\frac{k + (1 - \gamma)^2 (\epsilon^2 - \epsilon_{\operatorname{RL}})}{k}\right)\\
        & \label{eq:/b} \leq \frac{1}{2} \cdot \frac{(1 - \gamma)^2 (\epsilon^2 - \epsilon_{\operatorname{RL}})}{k}
    \end{align}
    where the first inequality is due to the fact that $(1 - \gamma)\sqrt{\epsilon_{\operatorname{RL}}/2}\geq 0$ and in the second inequality, we used $\ln(1 +x)\leq x$ for $x > -1$. Combining Equations~\ref{eq:a/b} and~\ref{eq:/b}, we obtain 
    \begin{align*}
        i \geq O\left(\frac{k}{(1 - \gamma)^2 (\epsilon^2 - \epsilon_{\operatorname{RL}})}\log \frac{k}{(1 - \gamma)^2\epsilon^2}\right). 
    \end{align*}
\end{proof}
\section{Approximate apprenticeship learning}\label{sec:approximate_AL}
In this section, we give an approximate classical algorithm for apprenticeship learning and present its convergence and time complexity analysis. We begin by introducing the classical subroutines that will be used in our algorithm. 

\subsection{Classical subroutines}\label{sec:classical_subroutines}
As noted in the work of~\cite{abbeel2004apprenticeship}, the optimization problem in Line~\ref{line:0SVM} of Algorithm~\ref{AL}
\begin{align}\label{eqn:SVM_problem}
    \displaystyle\argmax_{w:\Vert w\Vert_2\leq 1}\min_{j\in\{0, \cdots, (i-1)\}} w^T\left(\hat \mu_E - \mu^{(j)}\right)
\end{align}
can be solved using an SVM solver, where the approximate expert's feature expectations $\hat \mu_E$ are given the label 1 while feature expectations $\left\{\mu^{(\pi^{(j)})}:j = 0, \cdots, (i-1)\right\}$  are given the label -1. The vector $w$ that we seek is the unit vector orthogonal to the maximum margin separating hyperplane. 

The result below from~\cite{clarkson2012sublinear} returns an approximation of the vector $w$, the unit vector orthogonal to the maximum margin separating  hyperplane. 
\begin{fact}[\cite{clarkson2012sublinear}]\label{fact:SVM}
    There exists a classical algorithm that, given a data matrix $X\in\mathbb R^{n\times k}$, returns a vector $\bar w\in\mathbb R^k$ such that 
    \begin{align}\label{eqn:barw}
        X_i\bar w\geq \max_{w}\min_{i'\in[n]}X_{i'}w - \epsilon, \quad \forall i\in[n]
    \end{align}
    with probability at least 2/3, in $\tilde O\left(\frac{ n+k}{\epsilon^2} \right)$ time. 
\end{fact}

Recall that $\bar\mu^{(i)}\coloneqq \argmin_{\mu\in \operatorname{Co}\{\bar\mu^{(i)}, \mu^{(i+1)}\}} \left\Vert \hat\mu_E - \mu\right\Vert_2$. We show the following lemma on the Euclidean distance between $\bar w $ and the optimal solution $w^*$ of Fact~\ref{fact:SVM}. 
\begin{lemma}\label{lem:barw_w}
    Let $\epsilon\in (0, 1)$. Let $w^*$ be the optimal solution of Eq.(\ref{eqn:SVM_problem}) and let $\bar w$ be the approximate of $w^*$ computed by the classical algorithm in Fact~\ref{fact:SVM}. Then, the correctness guarantee of $\bar w$ in Eq.(\ref{eqn:barw}) implies that
    \begin{align*}
        \left\Vert \bar w - w^*\right\Vert_2\leq \sqrt{\frac{2\epsilon}{\Vert \hat\mu_E - \bar\mu^{(i)}\Vert_2}}. 
    \end{align*}
\end{lemma}
\begin{proof}
    Fix $i\in\mathbb Z_+$. Notice that Equation~\ref{eqn:SVM_problem} can be reformulated as the following optimization problem: 
    \begin{eqnarray}
        \nonumber\text{maximize} & & t\\
        \nonumber\text{subject to} & & w^T\left(\hat\mu_E  -\mu^{(j)}\right)\geq t, \hspace{1cm}\forall  j\in\{0, \cdots (i-1)\}. 
    \end{eqnarray}
    Let $w^*$ be the optimal solution of the above maximization problem. According to~\cite{abbeel2004apprenticeship}, $w^* = \frac{\hat\mu_E - \bar\mu^{(i)}}{\left\Vert \hat\mu_E - \bar\mu^{(i)} \right\Vert_2}$. Furthermore, the algorithm in Fact~\ref{fact:SVM} returns an approximation $\bar w$ of $w^*$ such that 
    \begin{align}\label{eqn:QSVM_result}
        \bar w^T\left(\hat\mu_E - \mu^{(j)}\right)\geq t - \epsilon, \quad \forall j\in\{0, \cdots, (i-1)\}. 
    \end{align}
    By definition of $w^*$, we have $\left(\hat\mu_E - \bar\mu^{(i)}\right) = t\cdot w^*$. Then, it follows that Eq.(\ref{eqn:QSVM_result}) implies
    \begin{align*}
        t\cdot \bar w^T w^*\geq t - \epsilon. 
    \end{align*}
    Dividing both sides of the inequality by $t$ gives 
    \begin{align}\label{eqn:IP_w}
        \bar w^T w^* \geq 1 - \frac{\epsilon}{t}. 
    \end{align}
    Then, 
    \begin{align*}
        \left\Vert \bar w - w^*\right\Vert_2^2 = \Vert \bar w\Vert_2^2 - 2 \displaystyle\sum_{i=1}^k \bar w_i\cdot w^*_i + \Vert  w^* \Vert_2^2 \leq 2 - 2\left(1- \frac{\epsilon}{t}\right) = \frac{2\epsilon}{t}, 
    \end{align*}
    where the inequality is due to Eq.(\ref{eqn:IP_w}) and the fact that $\left\Vert \bar w\right\Vert_2, \left\Vert w^*\right\Vert_2\leq 1$. Taking the square root on both sides of the inequality gives 
    \begin{align*}
        \left\Vert \bar w - w\right\Vert_2\leq \sqrt{\frac{2\epsilon}{t}}\leq \sqrt{\frac{2\epsilon}{\left\Vert \hat\mu_E - \bar\mu^{(i)}\right\Vert_2}}. 
    \end{align*}
\end{proof}

Classical multivariate Monte Carlo can be used to estimate feature expectations. The following is a result on multivariate Monte Carlo that assumes sampling access to the random variable whose mean we want to estimate. We rephrase this result to accuracy in the $\ell_2$ norm. 
\begin{fact}[Rephrased \citep{jerbi2022quantum}]\label{fact:ME}
    Let $\epsilon, \delta\in (0, 1)$. Let $v\in\mathbb R^k$ be such that $\Vert v\Vert_2\leq L$. Given sampling access to $v$, there exists a classical multivariate mean estimator that returns an estimate $\tilde \mu$ of $\mu = \mathbb E[v]$ such that $\Vert \tilde \mu - \mu\Vert_2\leq\epsilon$ with success probability at least $1 - \delta$ using $\tilde O \left(\frac{L^2k}{\epsilon^2}\right)$ samples of $v$. 
\end{fact}

The following reinforcement learning algorithm from~\cite{li2020breaking} returns an $\epsilon$-optimal policy in a generative model of sampling. In the generative model~\citep{kearns1998finite,kearns2002sparse,kakade2003sample}, one is allowed to choose an arbitrary state-action pair $(s, a)\in\mathcal S\times \mathcal A$ and request for a simulator to draw samples $s'\sim p(\cdot\vert s, a)$.  

\begin{fact}[\cite{li2020breaking}]\label{fact:RL}
    Let $\epsilon, \delta\in (0, 1)$. Let $\mathcal M = (\mathcal S, \mathcal A, R, P, \gamma)$ be a finite Markov decision process and let $\Phi\in[0, 1]^{SA\times k}$ be a features matrix. Given that $R = \Phi w$ for some $w\in\mathbb R^k$ such that $\Vert w\Vert_2\leq 1$ and assume access to $R$. Let $V^*$ and $V^{\tilde\pi}$ denote the value functions of the MDP when executing the optimal policy $\pi^*$ and $\epsilon$-optimal policy $\tilde\pi$ respectively. There exists a classical algorithm that returns an $\epsilon$-optimal policy $\tilde\pi$ such that $V^*(s) - \epsilon\leq V^{\tilde\pi}(s)\leq V^*(s)$ for all $s\in\mathcal S$ with probability at least $1 - \delta$ using $\tilde O\left(\frac{SA}{\epsilon^2(1 - \gamma)^{3}}\right)$ samples\footnote{The time complexity of the algorithm is the same as its sample complexity up to log factors, assuming that the generative model can be called in constant time.}. 
\end{fact}

We show in the following lemma that it suffices to truncate the MDP to $\tilde O\left(\frac{1}{1 - \gamma}\right)$ steps to obtain an $\frac{\epsilon}{2}$-error approximation of the feature expectation in the $\ell_2$-norm. Similar bounds are also shown in~\cite{abbeel2004apprenticeship, jerbi2022quantum}. 

\begin{lemma}\label{lem:truncation}
    Let $\epsilon\in (0, 1)$. Consider an MDP $(\mathcal S, \mathcal A, R, P, \gamma)$ with an infinite horizon $T = \infty$. It suffices to truncate the MDP to $H = \log_\gamma \left(\frac{\epsilon}{2}(1 - \gamma)\right)-1$ to obtain the following guarantee:
    \begin{align*}
        \left\Vert  \mathbb E\left[\displaystyle\sum_{t=0}^\infty \gamma^t \phi\left(s^{(t)}, a^{(t)}\right)\right] - \mathbb E\left[\displaystyle\sum_{t=0}^H \gamma^t \phi\left(s^{(t)}, a^{(t)}\right)\right]\right\Vert_2\leq\frac{\epsilon}{2}.
    \end{align*}
\end{lemma}
\begin{proof}
    First, note that we can bound 
    \begin{align}
        & \nonumber \left\Vert  \mathbb E\left[\displaystyle\sum_{t=0}^\infty \gamma^t \phi\left(s^{(t)}, a^{(t)}\right)\Bigg\vert\pi\right] - \mathbb E\left[\displaystyle\sum_{t=0}^H \gamma^t \phi\left(s^{(t)}, a^{(t)}\right)\Bigg\vert\pi\right]\right\Vert_2 \\
        & \nonumber = \left\Vert  \mathbb E\left[\displaystyle\sum_{t=0}^\infty \gamma^t \phi\left(s^{(t)}, a^{(t)}\right) - \displaystyle\sum_{t=0}^H \gamma^t \phi\left(s^{(t)}, a^{(t)}\right)\Bigg\vert\pi\right]\right\Vert_2\\
        & \nonumber = \left\Vert\mathbb E\left[\displaystyle\sum_{t=H+1}^\infty \gamma^t \phi\left(s^{(t)}, a^{(t)}\right)\big\vert \pi\right]\right\Vert_2\\
        \text{(Jensen's inequality)} & \leq \mathbb E\left[\left\Vert \displaystyle\sum_{t=H+1}^\infty \gamma^t \phi\left(s^{(t)}, a^{(t)}\right)\right\Vert_2\Bigg\vert \pi\right]\\ 
        & \label{eqn:second_part} \leq \sup_{s\in\mathcal S, a\in\mathcal A}\Vert \phi(s, a)\Vert_2 \left(\frac{\gamma^{H+1}}{1-\gamma}\right) \leq \left(\frac{\gamma^{H+1}}{1-\gamma}\right),
    \end{align}
    where the last inequality is due to the assumption that $\displaystyle\sup_{s\in\mathcal S, a\in\mathcal A}\Vert \phi(s, a)\Vert_2\leq 1$. It suffices to set $H = \log_\gamma \left(\frac{\epsilon}{2}(1 - \gamma)\right)-1$ to obtain the desired bound. 
\end{proof}

Using Lemma~\ref{lem:truncation}, we apply classical multivariate Monte Carlo to estimate feature expectations. 
\begin{lemma}\label{lem:ME}
    Let $\epsilon>0, \delta\in (0, 1)$ and $\gamma\in [0, 1)$ be a discount factor. Given access to the feature matrix $\Phi$, there exists a classical algorithm that outputs an estimate $\tilde \mu$ of $\mu = \mathbb E\left[\displaystyle\sum_{t=0}^\infty \gamma^t \phi\left(s^{(t)}, a^{(t)}\right)\Big\vert \pi\right]$ such that $\left\Vert \tilde\mu - \mu\right\Vert_2\leq \epsilon$  with success probability at least $1-\delta$ in time $\tilde O\left(\frac{k}{\epsilon^2(1 - \gamma)^3}\right)$. 
\end{lemma}
\begin{proof}
    By Lemma~\ref{lem:truncation}, observe the MDP trajectories $s^{(0)}, a^{0}, \cdots, s^{(H)}, a^{(H)}$. Then, we can estimate $\mathbb E\left[\displaystyle\sum_{t=0}^H \gamma^t \phi\left(s^{(t)}, a^{(t)}\right)\Big\vert \pi\right]$ up to additive error $\frac{\epsilon}{2}$ using Fact~\ref{fact:ME} with success probability at least $1 - \delta$. By triangle inequality, we obtain 
    \begin{align*}
         & \left\Vert \mathbb E\left[\displaystyle\sum_{t=0}^\infty \gamma^t \phi\left(s^{(t)}, a^{(t)}\right)\right] - \tilde\mu \right\Vert_2\\
         & \leq \left\Vert  \mathbb E\left[\displaystyle\sum_{t=0}^\infty \gamma^t \phi\left(s^{(t)}, a^{(t)}\right)\right] - \mathbb E\left[\displaystyle\sum_{t=0}^H \gamma^t \phi\left(s^{(t)}, a^{(t)}\right)\right]\right\Vert_2 + \left\Vert \mathbb E\left[\displaystyle\sum_{t=0}^H \gamma^t \phi\left(s^{(t)}, a^{(t)}\right)\right] - \tilde \mu\right\Vert_2 \\
         & \leq \frac{\epsilon}{2} + \frac{\epsilon}{2}\\
         & = \epsilon
    \end{align*}
    with success probability at least $1-\delta$. 

    We now analyze the time complexity. First, note that the feature matrix oracle is called at most $H = \log_\gamma \left(\frac{\epsilon}{2}(1 - \gamma)\right)-1 = \tilde O\left(\frac{1}{1 - \gamma}\right)$ times. Next, we compute the upper bound on the $\ell_2$-norm of the $k$-dimensional variable whose mean we wish to estimate. By the assumption on feature vectors, it follows that $\displaystyle\sup_{s\in\mathcal S, a\in\mathcal A}\Vert \phi(s, a)\Vert_2\leq 1$ and we have 
    \begin{align}\label{eqn:first_part}
        \left\Vert \displaystyle\sum_{t=0}^H \gamma^t \phi\left(s^{(t)}, a^{(t)}\right)\right\Vert_2  \leq  \displaystyle\sup_{s\in\mathcal S, a\in\mathcal A}\Vert \phi(s, a)\Vert_2 \cdot \left(\displaystyle\sum_{t=0}^H \gamma^t \right)\leq \frac{\left(1 - \gamma^{H+1}\right)}{1 - \gamma} \leq \frac{1}{1 - \gamma}. 
    \end{align}
    Plugging this quantity into the sample complexity of Fact~\ref{fact:ME}, we obtain a total running time\footnote{In this context, the sample complexity corresponds to the number of calls to the QMD.} of 
    \begin{align*}
        \tilde O\left(\frac{1}{1 - \gamma}\right) \cdot \tilde O \left(\frac{k}{\epsilon^2(1 - \gamma)^2}\right) = \tilde O \left(\frac{k}{\epsilon^2(1 - \gamma)^3}\right). 
       \end{align*}
\end{proof}

\subsection{Classical approximate apprenticeship learning algorithm}

Algorithm~\ref{AL} in the previous section finds policies such that their corresponding feature expectations converge to $\mu_E$. This is the ideal case when when $\hat\mu_E = \mu_E$ is a valid feature expectation vector and therefore lies in the convex hull of the set of feature expectations attained by optimal policies. In a more general case where $\hat\mu_E$ is a noisy estimate of $\mu_E$ and hence does not lie in the aforementioned convex hull, convergence to a small ball of radius $\rho\in\mathbb R_+$ centered around $\hat\mu_E$ that intersects the convex hull is considered.  This more robust algorithm is presented as the classical approximate apprenticeship learning algorithm in Algorithm~\ref{algo}. 

\begin{breakablealgorithm}
    \caption{Classical approximate apprenticeship learning algorithm}
    \label{algo}
    \begin{algorithmic}[1]
    \Require Initial policy $\tilde \pi^{(0)}$, errors $\epsilon, \epsilon_{\operatorname{RL}}\in (0, 1)$ such that $\epsilon\geq \sqrt{\epsilon_{\operatorname{RL}}}$, failure probability $\delta\in (0, 1)$. 
    \Ensure $\left\{\tilde \pi^{(i)}: i = 0, \cdots, n\right\}$. 
    \State \label{line:update0} Compute $\hat\mu_E$ in Equation~\ref{eqn:expert_feature_expectation} using Lemma~\ref{lem:ME} and update $\Phi'(1) = \hat\mu_E$. 
    \State \label{line:ME1} Obtain an estimate $\mu_q'^{(0)}$ of $\mu'^{(0)} \coloneqq \mu\left(\tilde\pi^{(0)}\right)$ with additive error $\frac{\epsilon}{3}$ and success probability $1 - \frac{\delta}{3n}$ using Lemma~\ref{lem:ME}. 
    \State \label{line:update1} Update $\Phi'(2)$ with $\hat \mu_E - \mu_q'^{(0)}$. 
    \State Set $i = 1$. 
    \State \label{line:SVM} Obtain an estimate $\bar w^{(i)}$ of $w^{(i)} = \argmax_{w:\Vert w\Vert_2\leq 1}\min_{j\in \{0, \cdots, (i-1)\}}w^T\left(\hat \mu_E - \mu^{(j)}\right)$ using Fact~\ref{fact:SVM} with error $\frac{\epsilon}{3}$ and success probability $1 - \frac{\delta}{3n}$. 
    \State \label{line:FindMin} Find $i_{\min} = \displaystyle\argmin_{j\in\{0, \cdots, (i-1)\}} \left\Vert \hat\mu_E - \mu_q'^{(j)}\right\Vert_2$.  
    \State \label{line:check} If $\tilde t^{(i_{\min})} = \left\Vert \hat\mu_E - \mu_q'^{(i_{\min})}\right\Vert_2\leq \epsilon + \frac{2\epsilon}{3} + \rho$ where $\rho=\frac{\epsilon}{3}$, 
    then terminate and set $n = i$. 
    \State \label{line:RL} Obtain an $\epsilon_{\operatorname{RL}}$-optimal policy $\tilde\pi^{(i)}$ using $\bar w^{(i)}$ and Fact~\ref{fact:RL} with success probability $1 - \frac{\delta}{3n}$.
    \State \label{line:ME2} Obtain an estimate $\mu_{q}^{(i)}$ of $\mu'^{(i)} \coloneqq \mu\left(\tilde \pi^{(i)}\right)$ with additive error $\frac{\epsilon}{3}$ and success probability $1 - \frac{\delta}{3n}$ using Lemma~\ref{lem:ME}. 
    \State \label{line:update2} Update $\Phi'(i+2)$ with $\hat\mu_E - \mu_q'^{(i)}$. 
    \State Set $i = i+1$ and go to Line~\ref{line:SVM}. 
\end{algorithmic}
\end{breakablealgorithm}

Before we proceed to discuss the convergence of Algorithm~\ref{algo}, we need the following result which shows that, if in every iteration $i$, $\bar{w}^{(i)}$ is an $\epsilon$-approximation for $w^{(i)}$ with respect to the notion of approximation as in Equation~\ref{eqn:barw}, then $\bar{w}^{(i)}$ obtained from Line~\ref{line:SVM} of Algorithm~\ref{algo} does not change the value of the policy too much. This allows us to combine the error from Line~\ref{line:SVM} with the other errors incurred in the algorithm. 
\begin{lemma}
    Let $\epsilon > 0$ and let $\mathcal M = (\mathcal S, \mathcal A, R, P, \gamma)$ be a Markov decision process with $R = \Phi w^*$ for some $w^*\in\mathbb R^k$ with $\Vert w^*\Vert_2\leq 1$, and let $\pi$ be the optimal policy of $\mathcal M$. Denote the value function with respect to $\pi$ under the reward $R$ by $V^{\pi}_R \in\mathbb R^S$. Furthermore, let $\bar w\in\mathbb R^k$ be related to $w^*$ as per Equation~\ref{eqn:barw}. Then, by replacing $R$ with $\bar R = \Phi\bar w$ and denoting the value function with respect to $\pi$ under the reward $\bar R$ as $V^{\pi}_{\bar R} \in\mathbb R^S$, we have 
    \begin{align*}
        \Vert V^{\pi}_{\bar R} - V^{\pi}_{R}\Vert_\infty\leq \epsilon. 
    \end{align*}
\end{lemma}
\begin{proof}
    First, note that 
    \begin{align*}
        \Vert V^{\pi}_{\bar R} - V^{\pi}_{R}\Vert_\infty
        & = \sup_{s\in\mathcal S} \vert \bar w\cdot \mu(\pi\vert s) - w ^*\cdot \mu(\pi\vert s) \vert \\
        & \leq \sup_{s\in\mathcal S} \Vert \bar w - w^*\Vert_2 \cdot \Vert \mu(\pi\vert s) \Vert_2\\
        \text{(Lem.~\ref{lem:barw_w})} & \leq \sqrt{\frac{2\epsilon_{\operatorname{SVM}}}{\left\Vert \hat\mu_E - \bar\mu_q'^{(i)}\right\Vert_2}} \cdot \sup_{s\in\mathcal S} \left\Vert \mathbb E\left[
        \left(\sum_{t=0}^{H} \gamma^t \phi(s^{(t)}, a^{(t)})+ \sum_{H+1}^\infty \phi(s^{(t)}, a^{(t)}) \right)\right]\right\Vert_2\\
        \text{(Jensen's inequality)} & \leq \sqrt{\frac{2\epsilon_{\operatorname{SVM}}}{\left\Vert \hat\mu_E - \bar\mu_q'^{(i)}\right\Vert_2}} \cdot  \sup_{s\in\mathcal S}\left\{ \mathbb E\left[ \left\Vert
        \left(\sum_{t=0}^{H} \gamma^t \phi(s^{(t)}, a^{(t)})+ \sum_{H+1}^\infty \phi(s^{(t)}, a^{(t)}) \right) \right\Vert_2 \right]\right\}\\
        \text{(Eqs.\ref{eqn:first_part},~\ref{eqn:second_part})} & \leq \sqrt{\frac{2\epsilon_{\operatorname{SVM}}}{\left\Vert \hat\mu_E - \bar\mu_q'^{(i)}\right\Vert_2}} \cdot  \left( \frac{1}{1 - \gamma} + \frac{\epsilon_{\operatorname{ME}}}{2}\right). 
    \end{align*}
   Having $\left\Vert \hat\mu_E - \bar\mu_q'^{(i)}\right\Vert_2 > 2\epsilon'$ (as long as Algorithm~\ref{algo} does not terminate) and setting $\epsilon_{\operatorname{SVM}} = \epsilon_{\operatorname{ME}} =  (1 - \gamma)^2 \epsilon'^3$, we have 
    \begin{align*}
        \Vert V^{\pi}_{\bar R} - V^{\pi}_{R}\Vert_\infty & < \sqrt{\frac{2(1 - \gamma)^2 \epsilon'^3}{2\epsilon'}} \cdot \left(\frac{1}{1 - \gamma} + \frac{(1 - \gamma)^2 \epsilon'^3}{2}\right)\\
        & = \epsilon' + \frac{(1 - \gamma)^3\epsilon'^4}{2}\\
        & \leq 2\epsilon'. \\
    \end{align*}
    Lastly, setting $\epsilon = 2\epsilon'$ yields the desired bound. 
\end{proof}

The supplementary material obtained through communication with the authors of~\cite{abbeel2004apprenticeship} contains a more general convergence proof of their apprenticeship learning algorithm, i.e. convergence to a $\rho$-radius ball around $\hat \mu_E$ to account for the fact that $\hat\mu_E$ could be a noisy estimate. They remarked that approximation errors from subroutines can be added to the radius of the ball around $\hat\mu_E$ to which the algorithm converges. Despite theses errors, the algorithm still converges after the same number of iterations. By this argument, Theorem~\ref{thm:iterations} remains true for Algorithm~\ref{algo} despite the fact that proof of Lemma~\ref{lem:lemma_3} holds for the ideal case where $\hat\mu_E\in\tilde M$. More specifically, the Algorithm~\ref{algo} still converges to 
a policy $\bar\pi$ such that $\left\Vert \hat\mu_E - \mu(\bar \pi)\right\Vert_2\leq \rho + \epsilon + \epsilon'$ after 
\begin{align*}
        O\left(\frac{k}{(1 - \gamma)^2 (\epsilon^2 - \epsilon_{\operatorname{RL}})}\log \frac{k}{(1 - \gamma)^2\epsilon^2}\right)
\end{align*}
iterations, where $\epsilon'$ denotes errors from other subroutines.

We now analyze the time complexity of one iteration of Algorithm~\ref{algo} in the following theorem. 
\begin{theorem}
    Let $\epsilon, \epsilon_{\operatorname{RL}}\in (0, 1)$ such that $\epsilon\geq \sqrt{\epsilon_{\operatorname{RL}}}$ and  $\gamma\in [0, 1)$. Then, a single iteration of Algorithm~\ref{algo} runs in time 
    \begin{align*}
        \tilde O\left(\frac{k + SA}{(1 - \gamma)^7\epsilon^6(\epsilon^2 - \epsilon_{\operatorname{RL}})}\right). 
    \end{align*}
\end{theorem}
\begin{proof}
    The running time of mean estimation in Lines~\ref{line:update0} and~\ref{line:ME2} is $\tilde O\left(\frac{k}{\epsilon^6(1 - \gamma)^7}\right)$. Updating $\Phi'$ in~ Lines\ref{line:update0},~\ref{line:update1} and~\ref{line:update2} takes $O(k)$ time. The time complexity of Line~\ref{line:SVM} is $O\left(\frac{n+k}{\epsilon^6(1 - \gamma)^4}\log n\right)$~\citep{clarkson2012sublinear} and finding the minimum in Line~\ref{line:FindMin} takes $O(n)$ time. Checking Line~\ref{line:check} takes $O(1)$ time. Lastly, Line~\ref{line:RL} can be implemented in time $\tilde O\left(\frac{SA}{\epsilon^2(1 - \gamma)^3}\right)$. This gives a per-iteration runtime of 
    \begin{align*}
        \tilde O\left(\frac{k}{\epsilon^6(1 - \gamma)^7}\right) + O(k) + \tilde O\left(\frac{n+k}{\epsilon^6 (1 - \gamma)^4}\right) + O(n) + \tilde O\left(\frac{SA}{\epsilon^2(1 - \gamma)^3}\right)
     \end{align*}
    By Theorem~\ref{thm:iterations}, we have $n= \tilde O\left(\frac{k}{(1 - \gamma)^2(\epsilon^2 - \epsilon_{\operatorname{RL}})}\right)$ for the number of iterations $n$. This brings the per-iteration time complexity of Algorithm~\ref{algo} to  
    \begin{align*}
        \tilde O\left(\frac{k + SA}{(1 - \gamma)^7\epsilon^6(\epsilon^2 - \epsilon_{\operatorname{RL}})}\right). 
    \end{align*}
\end{proof}

\section{Quantum algorithm}\label{sec:QAL}
In this section, we present the quantum algorithm for apprenticeship learning and its analysis. We begin by introducing the quantum subroutines that we will be using for our quantum algorithm. 

\subsection{Quantum subroutines}
We start with the first quantum subroutine, a quantum algorithm for finding the minimum by~\cite{durr1996quantum}.
\begin{fact}[\cite{durr1996quantum}]\label{fact:min_finding}
     Let $\delta\in (0, 1)$. Suppose we have a quantum access to a vector $v\in\mathbb R^k$ , there exists a quantum algorithm that finds $\displaystyle\min_{i\in[k]} v(i)$ with success probability at least $1-\delta$ in time $\tilde O\left(\sqrt k\log\frac{1}{\delta}\right)$. 
 \end{fact}

Quantum multivariate mean estimation allows one to estimate the mean of a multivariate random variable. We rephrase the following result from~\cite{cornelissen2022near} in terms of $\ell_2$-norm accuracy. 
\begin{fact}[\cite{zhong2023provably, cornelissen2022near}]\label{fact:QME}
    Let $X:\Omega\rightarrow\mathbb R^k$ be a $k$-dimensional bounded variable on the probability space $(\Omega, p)$ such that $\left\Vert X\right\Vert_2\leq L$ for some constant $L$. Assume access to a probability oracle $U_p:\ket{0}\rightarrow \displaystyle\sum_{\omega\in\Omega}\sqrt p(\omega)\ket{\omega}\ket{\psi_\omega}$ for some ancilla quantum states $\left\{\psi_\omega\right\}_{\omega\in\Omega}$ and a binary oracle $U_X:\ket{\omega}\ket{0}\rightarrow \ket{\omega}\ket{X(\omega)}$ for all $\omega\in\Omega$. Let $\epsilon > 0$ and $\delta\in (0, 1)$. There exists a quantum algorithm that outputs a estimate $\tilde\mu$ of $\mu = \mathbb E[X]$ such that $\left\Vert\tilde\mu\right\Vert_2\leq L$ and $\left\Vert \tilde\mu - \mu\right\Vert_2\leq\epsilon$ with probability at least $1-\delta$ using $O\left(\frac{L\sqrt k}{\epsilon}\log\frac{k}{\delta}\right)$ queries. 
\end{fact}

The next result is a quantum analogue of Line~\ref{fact:SVM}. It achieves a quadratic improvement in the dimension of the data and the size of the dataset but has a worse dependence on the error while having the same guarantees as Fact~\ref{fact:SVM}. 
\begin{fact}[\cite{li2019sublinear}]\label{fact:QSVM}
    Given a data matrix $X\in\mathbb R^{n\times k}$. There exists a quantum algorithm that returns a succinct classical representation of a vector $\bar w\in\mathbb R^k$ such that 
    \begin{align}
        X_i\bar w\geq \max_{w}\min_{i'\in[n]}X_iw - \epsilon, \quad \forall i\in[n]
    \end{align}
    with probability at least 2/3, using $\tilde O\left(\frac{\sqrt n}{\epsilon^4} + \frac{\sqrt k}{\epsilon^8}\right)$ quantum gates. 
\end{fact}
 Lastly, the result below is a quantum analogue of Fact~\ref{fact:RL}. Using the same (generative) model, there is an quadratic improvement in the cardinality of the action space, error and the effective time horizon.  
\begin{fact}[\cite{wang2021quantum}]\label{fact:QRL}
    Let $\epsilon, \delta\in (0, 1)$. Let $\mathcal M = (\mathcal S, \mathcal A, R, P, \gamma)$ be a finite Markov decision process and let $\Phi\in[0, 1]^{SA\times k}$ be a features matrix. Given that $R = \Phi w$ for some $w\in\mathbb R^k$ such that $\Vert w\Vert_2\leq 1$ and assume access to $R$. Let $V^*$ and $V^{\tilde\pi}$ denote the value functions of the MDP when executing the optimal policy $\pi^*$ and $\epsilon$-optimal policy $\tilde\pi$ respectively. There exists a quantum algorithm that returns an $\epsilon$-optimal policy $\tilde\pi$ such that $V^*(s) - \epsilon\leq V^{\tilde\pi}(s)\leq V^*(s)$ for all $s\in\mathcal S$ with probability at least $1 - \delta$ using $\tilde O\left(\frac{S \sqrt A}{\epsilon(1 - \gamma)^{1.5}}\right)$ samples\footnote{The time complexity of the algorithm is the same as its sample complexity up to log factors, assuming that the generative model can be called in constant time and access to a QRAM.}. 
\end{fact}

In the classical reinforcement learning environment, an agent interacts with the MDP by acting according to a policy $\pi$. Analogously, the quantum evaluation of a policy $\pi$ in quantum-accessible environments is given by the oracle $\mathcal O_{\pi}$ such that  
\begin{align*}
   \mathcal O_{\pi}:\ket{s}\ket{0}\rightarrow \sqrt{\pi(a\vert s)}\ket{s}\ket{a}. 
\end{align*}

Any policy that is classically computable can be efficiently converted into such a unitary~\citep{grover2002creating, jerbi2022quantum}. Moreover, a single call to oracles $\mathcal O_\pi$ takes constant time. 

We show that we can estimate feature expectations more efficiently than classical Monte Carlo. The following lemma combines ideas from classical sampling via quantum access (CSQA)~\citep{zhong2023provably} and quantum multivariate mean estimation~\citep{cornelissen2022near}. 

\begin{lemma}[Estimating feature expectations]\label{lem:QME}
    Let $H\in\mathbb Z_+$, $\epsilon>0, \delta\in (0, 1)$ and $\gamma\in [0, 1)$ be a discount factor. Given access to policy oracle $\mathcal O_{\pi}$, transition probability oracle $\mathcal O_{P}$ and a feature matrix oracle $\mathcal O_{\Phi}$, there exists a quantum algorithm that outputs an estimate $\tilde \mu$ of $\mu = \mathbb E\left[\displaystyle\sum_{t=0}^\infty \gamma^t \phi\left(s^{(t)}, a^{(t)}\right)\Big\vert \pi\right]$ such that $\left\Vert \tilde\mu - \mu\right\Vert_2\leq \epsilon$ and $\left\Vert \tilde\mu \right\Vert_2\leq \frac{1}{1 - \gamma}$ with success probability at least $1-\delta$ in time $\tilde O\left(\frac{\sqrt k}{\epsilon(1 - \gamma)^2}\log\frac{k}{\delta}\right)$. 
\end{lemma}
\begin{proof}
    We start by showing how to estimate $\mathbb E\left[\displaystyle\sum_{t=0}^H \gamma^t \phi\left(s^{(t)}, a^{(t)}\right)\right]$. We first prepare  
    \begin{align*}
        \ket{\psi^{(0)}}\coloneqq \displaystyle\sum_{s^{(0)}\in\mathcal S}\sqrt{p(s^{(0)})}\ket{s^{(0)}}\ket{\bar 0}. 
    \end{align*}
    Next, repeat the following for $t = 0, \cdots, H$:
    \begin{enumerate}
        \item Query $\mathcal O_{\pi}$ to obtain $\ket{\psi'}\coloneqq \mathcal O_{\pi}\ket{\psi^{(t)}}$. 
        \item Query $\mathcal O_{P}$ to obtain $\mathcal O_{P}\ket{\psi'}$ and collect the third register as $\ket{\psi^{(t+1)}}$. 
    \end{enumerate}
    The resulting state is 
    \begin{align*}
        \displaystyle\sum_{\begin{subarray}{l}s^{(0)}, \cdots,  s^{(H)}\in\mathcal S\\a^{(0)}, \cdots, a^{(H)}\in\mathcal A\end{subarray}}\sqrt{p\left(s^{(0)}, a^{(0)}, \cdots, s^{(H)}, a^{(H)}\right)}\ket{s^{(0)}, a^{(0)}}\cdots\ket{s^{(H)}, a^{(H)}}\ket{\bar 0}, 
    \end{align*}
    where $p\left(s^{(0)}, a^{(0)}\cdots, s^{(H)}, a^{(H)}\right)\coloneqq p\left(s^{(0)}\right) \Pi_{t=0}^{H-1}  p\left(s^{(t+1)}\vert s^{(t)}, a^{(t)}\right)\pi(a^{(t)}\vert s^{(t)})$ and $\ket{s^{(0)}, a^{(0)}}$ denotes $\ket{s^{(t)}}\ket{a^{(t)}}$ for $t\in\{0, \cdots, H\}$. For the rest of the proof, we will use $p$ to denote $p\left(s^{(0)}, a^{(0)}\cdots, s^{(H)}, a^{(H)}\right)$ for short. 
    
    Then, query $\mathcal O_{\Phi}$ on registers $2i-1$ and $2i$ for $i\in[H]$ to prepare the state  
    \begin{align*}
        \displaystyle\sum_{\begin{subarray}{l}s^{(0)}, \cdots, s^{(H)}\in\mathcal S\\a^{(0)}, \cdots, a^{(H)}\in\mathcal A\end{subarray}}\sqrt{p}\ket{s^{(0)}, a^{(0)}}\cdots\ket{s^{(H)}, a^{(H)}}\ket{\phi\left(s^{(0)}, a^{(0)}\right)}\cdots\ket{\phi\left(s^{(H)}, a^{(H)}\right)}\ket{\bar 0}. 
    \end{align*}
    Subsequently, we perform arithmetic computation to obtain 
    \begin{align*}
        \displaystyle\sum_{\begin{subarray}{l}s^{(0)}, \cdots , s^{(H)}\in\mathcal S\\a^{(0)}, \cdots, a^{(H)}\in\mathcal A\end{subarray}}\sqrt{p}\ket{s^{(0)}, a^{(0)}}\cdots\ket{s^{(H)}, a^{(H)}}\ket{\phi\left(s^{(0)}, a^{(0)}\right)}\cdots\ket{\phi\left(s^{(H)}, a^{(H)}\right)}\ket{\displaystyle\sum_{t=0}^H \gamma^t \phi\left(s^{(t)}, a^{(t)}\right)}. 
    \end{align*}
    Lastly, Fact~\ref{fact:QME} allows us to obtain an estimate $\tilde\mu$ of $\mu = \mathbb E\left[\displaystyle\sum_{t=0}^H \gamma^t \phi\left(s^{(t)}, a^{(t)}\right)\right]$ with additive error $\frac{\epsilon}{2}$ and success probability $1-\delta$. Therefore, by triangle inequality and Lemma~\ref{lem:truncation}, we obtain  
    \begin{align*}
         & \left\Vert \mathbb E\left[\displaystyle\sum_{t=0}^\infty \gamma^t \phi\left(s^{(t)}, a^{(t)}\right)\right] - \tilde\mu \right\Vert_2\\
         & \leq \left\Vert  \mathbb E\left[\displaystyle\sum_{t=0}^\infty \gamma^t \phi\left(s^{(t)}, a^{(t)}\right)\right] - \mathbb E\left[\displaystyle\sum_{t=0}^H \gamma^t \phi\left(s^{(t)}, a^{(t)}\right)\right]\right\Vert_2 + \left\Vert \mathbb E\left[\displaystyle\sum_{t=0}^H \gamma^t \phi\left(s^{(t)}, a^{(t)}\right)\right] - \tilde \mu\right\Vert_2 \\
         & \leq \frac{\epsilon}{2} + \frac{\epsilon}{2}\\
         & = \epsilon
    \end{align*}
    with success probability at least $1-\delta$. 
    
    We now analyze the time complexity. The policy oracle, transition probability oracle and feature matrix oracle are called at most $H = \log_\gamma \left(\frac{\epsilon}{2}(1 - \gamma)\right)-1 = \tilde O\left(\frac{1}{1 - \gamma}\right)$ times. By the same argument as Lemma~\ref{lem:ME},  the total running time is 
    \begin{align*}
        \tilde O\left(\frac{1}{1 - \gamma}\right) \cdot  O\left(\frac{\sqrt k}{\epsilon(1 - \gamma)}\log\frac{k}{\delta}\right) = \tilde O\left(\frac{\sqrt k}{\epsilon(1 - \gamma)^2}\log\frac{k}{\delta}\right). 
    \end{align*}
\end{proof}

\subsection{Quantum algorithm for apprenticeship learning}
Using the quantum subroutines and estimation of the feature vectors via CSQA from the previous subsection, we present our quantum algorithm for apprenticeship learning in Algorithm~\ref{QAL}. 
\begin{breakablealgorithm}
    \caption{Quantum apprenticeship learning algorithm}
    \label{QAL}
    \begin{algorithmic}[1]
    \Require Initial policy $\tilde\pi^{(0)}$, errors $\epsilon, \epsilon_{\operatorname{RL}}\in (0, 1)$ such that $\epsilon \geq \sqrt{\epsilon_{\operatorname{RL}}}$, failure probability $\delta\in (0, 1)$. 
    \Ensure $\left\{\tilde\pi^{(i)}: i = 0, \cdots, n\right\}$. 
    \State \label{line:Qupdate0} Compute $\hat\mu_E$ in Equation~\ref{eqn:expert_feature_expectation} using Lemma~\ref{lem:QME} and update $\Phi'(1) = \hat\mu_E$. 
    \State \label{line:QME1} Obtain an estimate $\mu_{q}'^{(0)}$ of $\mu'^{(0)} \coloneqq \mu\left(\tilde\pi^{(0)}\right)$ with additive error $\frac{\epsilon}{3}$ and success probability $1 - \frac{\delta}{4n}$ using~ Lemma\ref{lem:QME}. 
    \State \label{line:Qupdate1} Update $\Phi'(2)$ with $\hat \mu_E - \mu_q'^{(0)}$. 
    \State Set $i = 1$. 
    \State \label{line:QSVM} Obtain an estimate $\bar w^{(i)}$ of $w^{(i)}$ using~ Fact\ref{fact:QSVM} with error $\frac{\epsilon}{3}$ and success probability $1 - \frac{\delta}{3n}$. 
    \State \label{line:QMinFind} Find $i_{\min} = \displaystyle\argmin_{j\in\{0, \cdots, (i-1)\}} \left\Vert \hat\mu_E - \mu_q'^{(j)}\right\Vert_2$ using~ Fact\ref{fact:min_finding} with success probability $1 - \frac{\delta}{3n}$
    \State \label{line:Qcheck} If $\tilde t^{(i_{\min})} = \left\Vert \hat\mu_E - \mu_q'^{(i_{\min})}\right\Vert_2\leq \epsilon + \frac{2\epsilon}{3} + \rho$ where $\rho=\frac{\epsilon}{3}$, then terminate and set $n = i$. 
    \State \label{line:QRL} Obtain an $\epsilon_{\operatorname{RL}}$-optimal policy $\tilde\pi^{(i)}$ using $\bar w^{(i)}$ and Fact~\ref{fact:QRL} with success probability $1 - \frac{\delta}{3n}$.
    \State \label{line:QME2} Obtain an estimate $\mu_{q}'^{(i)}$ of $\mu'^{(i)} = \mu\left(\tilde \pi^{(i)}\right)$ with additive error $\frac{\epsilon}{3}$ and success probability $1 - \frac{\delta}{3n}$ using Fact~\ref{lem:QME}. 
    \State \label{line:Qupdate2} Update $\Phi'(i+2)$ with $\hat \mu_E - \mu_q'^{(i)}$. 
    \State Set $i = i+1$ and go to Line~\ref{line:QSVM}. 
\end{algorithmic}
\end{breakablealgorithm}

The convergence of Algorithm~\ref{QAL} is the same as that of Algorithm~\ref{algo} since both these algorithms have the same guarantees. The following theorem shows that Algorithm~\ref{QAL} is quadratically faster than Algorithm~\ref{algo} in terms of $k, A$ and $(1 - \gamma)$. 
\begin{theorem}
    Let $\epsilon, \epsilon_{\operatorname{RL}}\in (0, 1)$ such that $\epsilon\geq \sqrt{\epsilon_{\operatorname{RL}}}$ and $\gamma\in [0, 1)$. Then, a single iteration of Algorithm~\ref{QAL} runs in time 
    \begin{align*}
        \tilde O\left(\frac{\sqrt k + S\sqrt A}{(1 - \gamma)^{16}\epsilon^{24}(\epsilon^2 - \epsilon_{\operatorname{RL}})^{0.5}}\right). 
    \end{align*}
\end{theorem}
\begin{proof}
By Lemma~\ref{lem:QME}, Lines\ref{line:Qupdate0},\ref{line:QME1} and~\ref{line:QME2} can be implemented in $\tilde O\left(\frac{\sqrt k}{\epsilon^3(1 - \gamma)^4}\right)$ time. Line~\ref{line:QSVM} runs in $\tilde O\left(\frac{\sqrt n}{\epsilon^12(1 - \gamma)^8} + \frac{\sqrt k}{\epsilon^{24} (1 - \gamma)^{16}}\right)$ time while Line~\ref{line:QMinFind} can be implemented in time $\tilde O\left(\sqrt n\right)$. Checking Line~\ref{line:Qcheck} takes $O(1)$ time. Lastly, Line~\ref{line:QRL} can be implemented in time $\tilde O\left(\frac{S\sqrt A}{\epsilon(1 - \gamma)}\right)$. Therefore, the total time complexity for one iteration of Algorithm~\ref{QAL} is \begin{align*}
        \tilde O\left(\frac{\sqrt k}{\epsilon^3(1 - \gamma)^4}\right) + O\left(\frac{\sqrt n}{\epsilon^12(1 - \gamma)^8} + \frac{\sqrt k}{\epsilon^{24} (1 - \gamma)^{16}}\right) + \tilde O\left(\sqrt n\right) + \tilde O\left(\frac{S\sqrt A}{\epsilon(1 - \gamma)^{1.5}}\right). 
    \end{align*}
    By Theorem~\ref{thm:iterations}, Algorithm~\ref{QAL} terminates after $n = \tilde O\left(\frac{k}{(1 - \gamma)^2(\epsilon^2 - \epsilon_{\operatorname{RL}}) }\right)$ iterations. This brings the total per-iteration time complexity of Algorithm~\ref{QAL} to
    \begin{align*}
        \tilde O\left(\frac{\sqrt k + S\sqrt A}{(1 - \gamma)^{16}\epsilon^{24}(\epsilon^2 - \epsilon_{\operatorname{RL}})^{0.5}}\right). 
    \end{align*}
\end{proof}

\section{Conclusion}\label{sec:conclusion}
Based on the work of~\cite{abbeel2004apprenticeship}, we give a quantum algorithm for apprenticeship learning using inverse reinforcement learning. This apprenticeship learning framework of~\cite{abbeel2004apprenticeship} conveniently allows us to  use existing quantum subroutines for different parts of the algorithm, i.e. multivariate mean estimation, SVM solver, minimum finding and a reinforcement learning algorithm. As an intermediate step, we present a classical approximate apprenticeship learning algorithm using corresponding classical counterparts of the quantum subroutines. This is so to compare the speedup to which our quantum algorithm could achieve. 

We analyze the convergence of the apprenticeship learning algorithm by~\cite{abbeel2004apprenticeship} when a reinforcement learning subroutine that outputs nearly-optimal polices is used. We also perform time complexity analysis for both our classical approximate and quantum algorithms. Our results show that, as compared to the classical approximate algorithm, the quantum algorithm obtains a quadratic speedup in the dimension $k$ of the feature vectors and the size of the action space $A$, but suffers a worse error dependence $\epsilon$ and the effective time horizon $(1 - \gamma)$ due to the use of quantum SVM solver subroutine and the tuning of error to achieve convergence. A similar phenomenon in which the time complexity is quadratically improved in terms of  certain parameters while the dependence on other parameters suffers a worse scaling has occurred in the context of quantum algorithms for semidefinite programming (SDP)~\citep{brandao2016quantum}. In their work, the quantum algorithm for SDP achieve a quadratic speedup in the problem dimension but has a worse dependence on other parameters. Nevertheless, both our algorithms have the same convergence guarantees. Specifically, we proved that the algorithms terminate after $ \tilde O\left(\frac{k}{(1 - \gamma)^2 (\epsilon^2 - \epsilon_{\operatorname{RL}})} \right)$ iterations. 

Assuming the use of QMD, the depth complexity of our quantum algorithm is the same as its time complexity, up to logarithmic factors. The QMD is a newly introduced formalization of data access models in the literature. It is a generalization of the popular QRAM and QRAG, whose limitations and potentials in real-world implementations have been studied by~\cite{clader2023quantum, giovannetti2008quantum, giovannetti2008architectures, di2020fault, hann2021practicality, phalak2023quantum, jaques2023qram, allcock2023constant, luongo2022quantum} and~\cite{buhrman2022memory, luongo2022quantum, allcock2023constant} respectively. The idea of using a QMD is so that we can delegate the resources to the implementation of the data access via QMD. This computational model can also be augmented with algorithms for efficiently moving and addressing qubits within physically realistic devices~\citep{beals2013efficient}.

In our work, we assume that the reward function is expressible as a linear function of known features. One possible future direction would be to consider apprenticeship learning in a setting where the reward function is a nonlinear function of feature vectors. It would also be interesting to apply our quantum algorithm as a subroutine to solve learning problems such as the Hamiltonian learning problem, where known results in the literature could be used as expert's policy. By cleverly mapping parameters of the learning problem to MDP parameters, one could design a quantum apprenticeship learning algorithm for learning quantum systems.


\backmatter

\bmhead{Acknowledgements}
DL specially thanks Pieter Abbeel and Andrew Y. Ng for providing supplementary material for the paper ``Apprenticeship Learning via Inverse Reinforcement Learning". DL thanks Patrick Rebentrost and Alessandro Luongo for helpful discussions and  pointing out References~\cite{li2019sublinear} and~\cite{allcock2023constant} respectively. AA and DL gratefully acknowledge funding from the Latvian Quantum Initiative under EU Recovery and Resilience Facility under project no.\ 2.3.1.1.i.0/1/22/I/CFLA/001.

\bibliography{sn-bibliography}





a

\end{document}